\documentclass[lettersize,journal,twoside]{IEEEtran}
\usepackage{amsmath,amsfonts}
\usepackage{array}
\usepackage{textcomp}
\usepackage{stfloats}
\usepackage{url}
\usepackage{verbatim}
\usepackage{cite}
\usepackage{cite}
\usepackage{amssymb}
\usepackage{amsfonts}
\usepackage{amsmath}
\usepackage{amsthm}
\usepackage{mathrsfs}
\usepackage{bm}
\usepackage{verbatim}
\usepackage{autobreak} 
\usepackage{graphicx} 
\usepackage{float} 
\usepackage{subfigure} 
\usepackage{tabularx}
\usepackage{algpseudocode} 
\usepackage{enumerate}
\usepackage{cases}
\usepackage[linesnumbered, ruled]{algorithm2e}

\usepackage{url}
\usepackage[colorlinks,
            linkcolor=black,
            anchorcolor=blue,
            citecolor=blue]{hyperref}
\usepackage{hyperref}
\newtheorem{lemma}{Lemma}

\newtheorem{proposition}{Proposition}

\begin{document}

\title{Joint Beamforming Optimization and Mode Selection for RDARS-Aided MIMO Systems}

\author{ {Jintao Wang,   Chengzhi Ma, Shiqi Gong, Xi Yang, Shaodan Ma, \textit{{Senior Member}}, \textit{{IEEE}} }
\thanks{ J. Wang, C. Ma, and S.  Ma are with the State Key Laboratory of Internet of Things for Smart City and the Department of Electrical and Computer Engineering, University of Macau, Macao SAR, China (e-mails: wang.jintao@connect.um.edu.mo; yc07499@um.edu.mo; shaodanma@um.edu.mo). }
\thanks{ S. Gong is with the School of Cyberspace Science and Technology, Beijing Institute of Technology, Beijing 100081, China (e-mail: gsqyx@163.com).}
\thanks{X. Yang is with the Shanghai Key Laboratory of Multidimensional Information Processing, School of Communication and Electronic Engineering, East China Normal University, Shanghai 200241, China (email: xyang@cee.ecnu.edu.cn).}
}

\markboth{IEEE Transactions on Wireless Communications}%
{Wang \MakeLowercase{\textit{et al.}}: Joint Beamforming Optimization and Mode Selection for RDARS-aided MIMO Systems}


\maketitle

\IEEEpubid{\begin{minipage}[t]{\textwidth}\ \\[12pt] \centering
  \copyright \ 2024 IEEE. Personal use of this material is permitted. Permission from IEEE must be obtained for all other uses, in any current or future media, including reprinting/republishing this material for advertising or promotional purposes, creating new collective works, for resale or redistribution to servers or lists, or reuse of any copyrighted component of this work in other works.
\end{minipage}} 

\begin{abstract}
Reconfigurable intelligent surface (RIS) has emerged as a cost-effective solution for green communications in 6G. However, its further extensive use has been greatly limited due to its fully passive characteristics. 
Considering the appealing distribution gains of distributed antenna systems (DAS), a flexible reconfigurable architecture called reconfigurable distributed antenna and reflecting surface (RDARS) is proposed. 
RDARS encompasses DAS and RIS as two special cases and maintains the advantages of distributed antennas while reducing the hardware cost by replacing some active antennas with low-cost passive reflecting surfaces.
In this paper, we present a RDARS-aided uplink multi-user communication system and investigate the system transmission reliability with the newly proposed architecture. 
Specifically, in addition to the distribution gain and the reflection gain provided by the connection and reflection modes, respectively, we also consider the dynamic mode switching of each element which introduces an additional degree of freedom (DoF) and thus results in a selection gain.
As such,  we aim to minimize the total sum mean-square-error (MSE) of all data streams by jointly optimizing the receive beamforming matrix, the reflection phase shifts and the channel-aware placement of elements in the connection mode. 
To tackle this nonconvex problem with intractable binary and cardinality constraints, we propose an inexact block coordinate descent (BCD) based penalty dual decomposition (PDD) algorithm with the guaranteed convergence.
Since the PDD algorithm usually suffers from high computational complexity, a low-complexity greedy-search-based alternating optimization (AO) algorithm is developed to yield a semi-closed-form solution with acceptable performance.
Numerical results demonstrate the superiority of the proposed architecture compared to the conventional fully passive RIS or DAS. Furthermore, some insights about the practical implementation of RDARS are provided.  
\end{abstract}


\section{Introduction}
As a revolutionary technology aimed at realizing green communication, reconfigurable intelligent surface (RIS) plays a pivotal role in the development of 6G.
Unlike conventional technologies designed to enhance wireless communication, RIS exhibits remarkable capability to dynamically reshape the wireless propagation environment and explore new physical dimensions of transmission.
By deploying a massive number of low-cost passive elements, RIS can adjust the phase of the incident signal to create favorable wireless channels, thereby improving communication performance.

Due to its high potential for energy and spectrum efficiency, numerous research works have been conducted to explore the benefits of RIS-aided communication systems.
These studies have shown RIS has the ability to address a wide range of challenges, such as improving coverage area, overcoming local dead zones, enhancing edge user rates,  eliminating co-channel interference, and enabling high-precision positioning \cite{Chen2022RISsurvey,2021Basharat}.
For instance, the authors in \cite{ozdogan2020using} optimized RIS to enrich the propagation environment by adding multi-paths and increasing the rank of the channel matrix to support spatial multiplexing, thus leading to substantial capacity gains.
The authors in \cite{CoverageEnhan2022} conducted a performance analysis of the coverage enhancement for RIS under current commercial mobile networks in various urban scenarios. The field trial results showed a significant improvement in the user experience, including extending the coverage and improving the throughput. 
However, due to the fully passive nature of RIS's elements, its further extensive use has been greatly limited. 
For example, unless a large number of elements are deployed, the performance improvement brought by RIS is limited due to the ``multiplicative fading" effect, particularly in the presence of a direct link. Moreover, the signaling overhead for channel estimation and the configuration overhead for passive beamforming increase with the number of reflecting elements, thereby restricting its implementation to some extent.
Although some RIS variants, e.g., active RIS \cite{2active2023,Dong2022ARIS}, semi-passive RIS \cite{kang2023active, Peng2023active, Sankar2023active } or hybrid relay-reflecting intelligent surface (HR-RIS) \cite{nguyenHybridRelayReflectingIntelligent2022, Nguyen2023Spectral} have been proposed to address some of the challenges above. The incorporation of extra reflection-type power amplifiers (PAs) or radio-frequency (RF) chains may fail to meet the ultra-low-power requirements in 6G. 

Another promising technology aimed at improving spectral efficiency in 6G networks is distributed antenna system (DAS) \cite{Moerman2022DAS,Yu2020DAS}.
DAS involves the deployment of a large number of spatially distributed antennas or radio access points (RAPs) to provide uniform services to users. Various DAS architectures have been proposed, including cloud radio access network (C-RAN), ultra-dense network (UDN), and cell-free massive multiple-input multiple-output (CF-mMIMO) systems.
For example, CF-mMIMO, where a large number of distributed access points are connected to a central processing unit and serve all users within a network, has a high potential in improving the network performance from different perspectives by integrating emerging technologies such as non-orthogonal multiple access (NOMA) and unmanned aerial vehicle (UAV) \cite{SE2021DASsurvey}. 
However, such dense deployment of a large number of antennas or RAPs leads to significant hardware costs and increased energy consumption.

Recently, a novel architecture called reconfigurable distributed antennas and reflecting surfaces (RDARS) has been proposed to maintain the benefits of distributed antennas while reducing the hardware cost by replacing some active antennas with low-cost passive reflecting surfaces\cite{ma2023reconfigurable}.
Specifically, RDARS combines the flexibility of distributed antennas and reconfigurability of passive reflecting surfaces, where each element can be dynamically programmed to two modes: the reflection mode and the connection mode. Elements in the reflection mode function similarly to the conventional passive RIS, while elements in the connection mode act as distributed antennas to directly transmit/receive signals.
The proposed RDARS architecture inherits the low-cost advantage of the conventional RIS while overcoming the multiplicative fading effect. By utilizing the connection mode, RDARS achieves additional distribution gains. Furthermore, elements in the connection mode also facilitate channel estimation of the RDARS-related channels and ease the RDARS configuration with reduced overhead. The RDARS architecture demonstrates significant superiority over conventional RIS and DAS, as highlighted by the theoretical performance analysis and experimental demonstrations \cite{ma2023reconfigurable}. Additionally, a RDARS-aided integrated sensing and communication system (ISAC) prototype was implemented in \cite{wang2023reconfigurable}, showcasing the potential of RDARS in future wireless systems.

{\color{black} 
There are several advantages when implementing RDARS.
Firstly, RDARS can be deployed by building upon existing distributed antenna systems, such as C-RAN, UDN, and CF-mMIMO systems. Such integration and reuse of the infrastructure significantly reduce the deployment costs and complexities.
From the perspective of RISs, they are interfaced with the base station (BS) through a smart controller. The controller intelligently adjusts the phase shifters of the RIS to optimize wireless channels. Due to the fully passive characteristics, the optimization process must be performed at the BS.
Various methods can be used for controlling RIS, including wired connections, IP routing, wireless connections, or autonomous sensing \cite{RISwhitepaper2022}. While wireless controls offer flexibility in the deployment and power consumption, achieving high-accuracy beamforming through wireless controls is still challenging. Additionally, wireless controls require precise time synchronization between the BS and RIS, which is difficult to implement in practice. Moreover, mutual interference between wireless access links and wireless controls can impact the overall network performance.
On the other hand, the distributed antennas are connected to a central unit through optical front-hauls, including baseband-over-fiber (BBoF), intermediate-frequency-over-fiber (IFoF), and analog radio-frequency-over-fiber (RFoF) links \cite{Yu2020DAS}. The implementation and configuration of RDARS can leverage the existing optical infrastructure, thus enhancing the practicality of RDARS-aided systems.
By considering these practical aspects, the implementation of RDARS becomes feasible and can be aligned with current technological capabilities.

}

To the best of our knowledge, the potential of such appealing architecture has not been fully investigated. 
Due to the channel randomness, the location of elements in the connection mode also influences the system performance. 
The RDARS structure offers the potential to achieve selection gains by leveraging the additional degree of freedom (DoF) provided by the dynamically programmed modes for all elements.
Therefore, considering the appealing structure of RDARS with distribution, reflection and selection gains, we study the RDARS-aided multi-user uplink system and investigate the potential of RDARS on transmission reliability.   
In this paper, we aim to minimize the total mean-square-error (MSE) of all data streams by joint beamforming design and channel-aware placement of elements in the connection mode characterized by a binary selection matrix. 
Since the channel-aware placement involves binary variables and unit-modulus constraints, the MSE minimization becomes the mixed-integer programming problem, which is non-deterministic polynomial-time (NP)-hard and lacks efficient solutions.
To tackle such a non-convex combinatorial optimization problem, we propose an inexact block coordinate descent (BCD)-based penalty dual decomposition (PDD) algorithm to alternatively optimize the beamforming and selection variables. Furthermore, a greedy-search-based alternating optimization (AO) algorithm is proposed to reduce the computational complexity with semi-closed-form solutions. Moreover, we discover some insightful results under special channel conditions.
Simulation results demonstrate the superiority of the proposed RDARS-aided system compared to the conventional fully passive RIS-aided and DAS systems. 

The main contributions are summarized as follows:
\begin{itemize}
    \item To our best knowledge, this paper presents the first investigation of a novel RDARS-aided system that incorporates distribution gains, reflection gains, and selection gains. 
    We propose a RDARS-aided uplink multi-user communication system and explore the performance improvement in transmission reliability brought by the newly proposed architecture. 
    \item We consider the flexible and programmable configuration of the RDARS's elements, where each element can be dynamically switched between the connection and reflection modes as needed. 
    The additional DoF for dynamic mode switching is investigated to enhance the system performance.
    To improve the system transmission reliability, the total MSE of all data streams is minimized by jointly optimizing the receive beamforming matrix, the reflection phase shifts, and the channel-aware placement of elements in the connection mode.
    \item To tackle the intractable binary and cardinality constraints, we first propose an inexact BCD-based PDD algorithm to alternatively optimize the beamforming vector and the mode selection matrix. Furthermore, in order to reduce the computational complexity, we approximate the original objective function and propose a low-complexity greedy-search-based AO algorithm, which yields a semi-closed-form solution with acceptable performance.  Moreover, we present some insightful results in the special cases of the RDARS-aided single-user scenario.
    \item The performance comparison between the RDARS and the conventional fully passive RIS or DAS systems is conducted through simulations. The results demonstrate the superiority of the proposed architecture. It also shows the necessity of optimizing the location of elements in the connection mode according to the channel conditions.  We further analyze the impact of the number of elements in the connection mode, thereby offering valuable insights for the practical implementation of RDARS. 
\end{itemize}

The remainder of this paper is organized as follows.  Section II introduces the RDARS's architecture, system model, and problem formulation. The proposed inexact BCD-based PDD algorithm is presented in Section III. Section IV provides a low-complexity AO algorithm. Numerical results are shown in Section V. Finally, Section VI concludes this paper. 

    \begin{figure}[t]
        \centering  
     \includegraphics[width=0.45\textwidth]{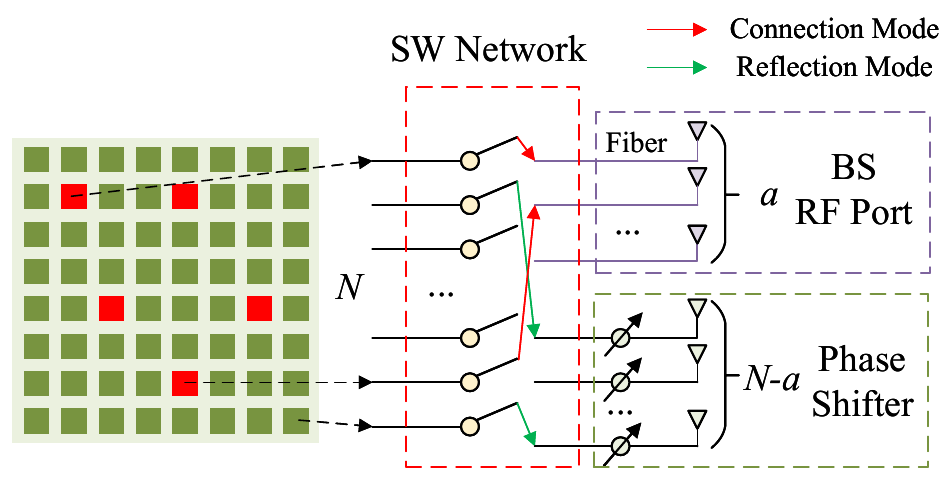}
        \caption{Illustration of the dynamic mode switching of RDARS.}
        \vspace{-6pt}
        \label{SW Network}    
    \end{figure}
    
Notation: 
Through this paper, scalars, vectors, and matrices are denoted by lowercase, bold-face lowercase, and bold-face uppercase letters, respectively.  
The notation $\mathbb{C}^n$ denotes the $n$-dimensional space of the complex number.
$(\cdot)^T$ and $(\cdot)^H$ represent the transpose and conjugate transpose of a complex-valued vector or matrix, respectively.  
${\bf{I}}_{N}$ denotes the identity matrix of size $N \times N$. 
The expectation operator is denoted by $\mathbb{E}[\cdot]$. 
${\rm diag}({\bf{a}})$ indicates a square diagonal matrix whose diagonal elements consist of a vector ${\bf{a}}$, and ${\rm diag}({\bf{A}})$ returns a vector consisting of the diagonal elements of matrix ${\bf{A}}$. 
${\rm \widetilde{diag}}({\bf{A}})$ represents a square diagonal matrix whose diagonal elements are the same as those of matrix ${\bf{A}}$.
${\bf{B}} = {\rm blkdiag}({\bf{A}}_1,...,{\bf{A}}_N)$ returns the block diagonal matrix created by aligning the input matrices ${\bf{A}}_1$,...,${\bf{A}}_N$ along the diagonal of ${\bf{B}}$.
$|a|$ and $\angle a$ refer to the modulus and phase of a complex input $a$, respectively. 
$||a||_0$ denotes the zero norm of vector $a$.
 $ \mathcal{CN}({\bf{0}},{\bf{A}})$ means the distribution of a circularly symmetric complex Gaussian (CSCG) random vector with zero mean vector and covariance matrix ${\bf{A}}$ and $\sim$ denotes ``distributed as''.
 $\Re\{\cdot \}$ returns the real part of the complex input. 
 
\section{System Model and Problem Formulation}

\subsection{RDARS Architecture}    
A RDARS composed of $N$ elements is employed, with $a$ elements connected to the BS and operating in the ``connection mode", while the remaining $N\!-\!a$ passive elements operate in the ``reflection mode" to reflect the incident signal to favorable directions.  
Since each element can be dynamically switched between the connection and reflection modes, it offers an extra DoF to achieve selection gains.
Such transformation can be realized via a switching network (SW) as illustrated in Fig.~\ref{SW Network} \footnote{ \color{black}
Compared to conventional DAS and RIS, the additional hardware in the RDARS is the SW, which includes RF switches designed for dynamic mode switching. As we know, RF switches have been utilized in antenna selection techniques within the domain of hybrid beamforming for massive MIMO communications. They serve to reduce hardware costs by substituting for analog phase shifters or RF chains \cite{sanayei2004antenna,Gao2018Hybrid,elbir2019joint}. Analog RF switches offer significantly lower complexity and power consumption, which enhances the appeal of antenna selection \cite{Ant-SW-R1,Ant-SW-R2}. 
Although the cable connections and mode switches involved in the hardware implementation of RDARS will incur an acceptable hardware cost, the performance improvement makes the RDARS structure particularly appealing.}.
Flexible configurations are available for the number and locations of elements in the connection mode.
In this paper, we assume the number of connected elements, i.e., $a$, is predefined.

In comparison to the traditional fully passive RIS, the RDARS architecture inherits low-cost advantages from elements in the reflection mode. On the other hand, it overcomes the multiplicative fading effect from those in the connection mode. 
Since a small number of elements in the connection mode can lead to a substantial performance improvement, we assume the number of connected elements is much smaller than the total number of elements on RDARS, i.e., $a<<N$.

Denoted by $\mathcal{A}$ the index set of the connected components, the mode selection matrix can be represented by the diagonal matrix ${\bf{A}}\in \mathbb{C}^{N \times N}$, whose diagonal value equal to 1 for the connection mode and 0 for the reflection mode. 
The reflection coefficients of all elements operating in the reflection mode can be denoted as $ {\bm{\theta}}\in \mathbb{C}^{N \times 1}$ with ${\bf{\Phi}}\triangleq{\rm diag}({\bm{\theta}})$, while $({\bf{I}}\!-\!{\bf{A}}){\bm{\theta} }$ stands for the realistic phase shifts of elements in the reflection mode. 
 
    \begin{figure}[t]
        \centering  
     \includegraphics[width=0.45\textwidth]{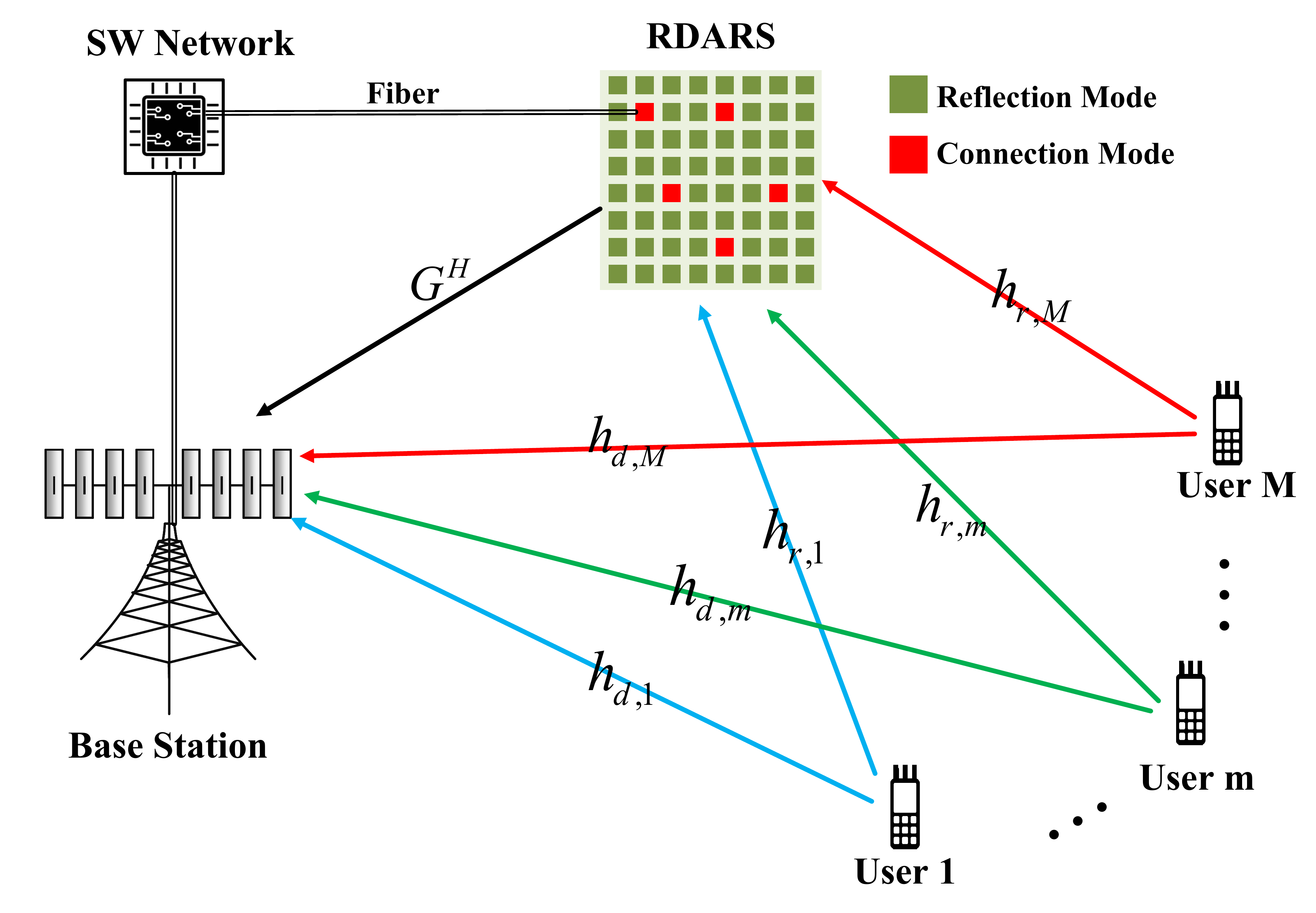}
        \caption{A RDARS-aided uplink MIMO communication system.}
        \vspace{-6pt}
        \label{system model}    
    \end{figure} 
    
\subsection{System Model}
As shown in Fig.~\ref{system model}, a RDARS-aided uplink multi-user communication system with $M$ single-antenna users and a BS equipped with $N_r$ antennas is considered. 
Considering the quasi-static far-field channel environment,  the channel from $M$ users to the BS and RDARS can be denoted as ${\bf{H}}_{d} \!\in\!\mathbb{C}^{N_r \times M}$ and ${\bf{H}}_{r} \!\! \in\!\! \mathbb{C}^{N \times M}$, respectively, where
${\bf{H}}_{d}=\big[ {\bf{h}}_{d,1} ~  {\bf{h}}_{d,2} \dots {\bf{h}}_{d,M} \big]$ and ${\bf{H}}_{r}=\big[ {\bf{h}}_{r,1} ~  {\bf{h}}_{r,2} \dots {\bf{h}}_{r,M} \big]$.
In addition, the channel between the BS and the RDARS is represented as ${\bf{G}} \in \mathbb{C}^{N \times N_r}$. 
All channel state information (CSI) is assumed to be known perfectly to explore the performance limit\footnote{Since the elements on RDARS can be dynamically programmed to the connection mode, the channel related to the connected elements can be easily estimated via the pilot training. Recalling the channel estimation methods in asymmetrical uplink and downlink massive MIMO systems \cite{Yang2023antselect}, the full-dimensional CSI related to the RDARS can be efficiently recovered.}.
Denoted by ${\bf{s}}$ the transmitted symbols from $M$ users with $\mathbb{E}[{\bf{s}} {\bf{s}}^H]={\bf{I}}$, the uplink received signal ${\bf{y}}\in \mathbb{C}^{(N_r+a) \times 1}$ at the BS can be divided into two parts, i.e., ${\bf{y}}^{b}\in \mathbb{C}^{N_r \times 1}$ and ${\bf{y}}^{c}\in \mathbb{C}^{a \times 1}$, defined as follows\footnote{We assume the ideal inter-element isolation of the practical arrays on RDARS, where the received signal originating from the remaining elements in the reflection and connection modes can be ignored \cite{2active2023,Hu2021semiRIS}. }

\vspace{-4pt}
\begin{equation}
  \underbrace{\begin{bmatrix} {\bf{y}}^{b} \\ {\bf{y}}^{c} \end{bmatrix}}_{{\bf{y}}}
  = 
  \underbrace{\begin{bmatrix} {\bf{H}}_{d} \!+\!{\bf{G}}^H ({\bf{I}}\!-\!{\bf{A}}) {\bm{\Phi}} {\bf{H}}_{r} \\  {\bf{A}}_a {\bf{H}}_{r}  \end{bmatrix}   }_{ {\bf{ H}} } {\bf{P}}{\bf{s}} + \underbrace{\begin{bmatrix} {\bf{n}}_{b} \\  {\bf{n}}_{c} \end{bmatrix} }_{{\bf{n}}},
\end{equation} 
where ${\bf{y}}^{b}$ and ${\bf{y}}^{c}$ represent the signals received by the BS's antennas and the RDARS's elements operating in the connection mode, respectively.  
The effective channel ${\bf{H}}$ is composed of two parts, i.e., ${\bf{H}}_b$ and ${\bf{H}}_c$. 
First, ${\bf{H}}_b$ represents the effective channel between users and the BS assisted with the RDARS's elements in the reflection mode, similar to that of the conventional RIS-aided system, defined as ${\bf{H}}_b={\bf{H}}_{d}+{\bf{G}}^H ({\bf{I}}\!-\!{\bf{A}}) {\bf{\Phi}} {\bf{H}}_{r}$.
On the other hand, ${\bf{ H}}_c$ denotes the effective channel between the users and the elements working in the connection mode, i.e., ${\bf{H}}_c={\bf{A}}_a  {\bf{H}}_{r}$, where ${\bf{A}}_a \in \mathbb{C}^{a \times N}$ stands for the channel index matrix of those components connected to the BS. Specifically, ${\bf{A}}_a$ is a submatrix of ${\bf{A}}$ consisting of non-zero rows of the mode selection matrix ${\bf{A}}$.
Here, ${\bf{P}}$ denotes the transmit power matrix for $M$ users, i.e., ${\bf{P}}={\rm diag}([\sqrt{p_1}, \sqrt{p_2}, \dots, \sqrt{p_M}])$, with $p_m$ representing the transmit power for the $m$-th user. 
 ${\bf{n}}_b$ and ${\bf{n}}_c$ represent the additive white Gaussian noise (AWGN) obeying the CSCG distributions with zero mean and covariance matrices ${\bm{\Lambda}}_b \!=\! {\sigma_b^2}{\bf{I}}_{N_r}$ and ${\bm{\Lambda}}_c \!=\!{\sigma_c^2}{\bf{I}}_{a}$, i.e., ${\bf{n}}_b \sim \mathcal{CN}({\bf{0}},{\sigma_b^2}{\bf{I}}_{N_r})$ and ${\bf{n}}_c \sim \mathcal{CN}({\bf{0}},{\sigma_c^2}{\bf{I}}_{a})$, respectively. 
${\bm{\Lambda}}={\rm blkdiag} { ({\bm{\Lambda}}_b , {\bm{\Lambda}}_c) }$ is the covariance matrix of the effective noise vector ${\bf{n}}$.
{\color{black}For simplicity, we assume ${\sigma_b^2}={\sigma_c^2}={\sigma^2}$. Thus, the covariance matrix of the effective noise vector turns into ${\bm{\Lambda}}={\sigma^2}{\bf{I}}$.}

At the BS side, the receive beamforming matrix ${\bf{W}}\in\mathbb{C}^{M \times (N_r+a)} $ is applied to estimate the data symbol vector ${\bf{\hat s}}={\bf{W}} {\bf{y}} $. The resultant MSE of individual data streams of all users can be denoted as diagonal elements of the average MSE matrix derived in \eqref{Ori_MSE} at the top of this page, 
\begin{figure*}
\vspace{-12pt}
\begin{small}
\begin{align} \label{Ori_MSE}
  {{\rm {\bf MSE}}}
  = & \mathbb{E}[({\bf{\hat s}}-{\bf{s}})({\bf{\hat s}}-{\bf{s}})^H] 
  =  \underbrace{{\bf{W}}_b ( {\bf{H}}_b {\bf{P}} {\bf{P}}^H{\bf{H}}_b^H \!+\!  \sigma_b^2 {\bf{I}}_{N_r}){\bf{W}}_b^H \!-\! {\bf{W}}_b {\bf{H}}_b {\bf{P}} \!-\! {\bf{P}}^H {\bf{H}}_b^H {\bf{W}}_b^H \!+\! {\bf{I}}_M }_{{\bf{E}}^{\rm RIS}}   \\
  &+ \underbrace{ {\bf{W}}_c ( {\bf{H}}_c {\bf{P}} {\bf{P}}^H{\bf{H}}_c^H \!+\!  \sigma_c^2 {\bf{I}}_{a}){\bf{W}}_c^H \!-\! {\bf{W}}_c {\bf{H}}_c {\bf{P}} - {\bf{P}}^H {\bf{H}}_c^H {\bf{W}}_c^H \!+\! {\bf{I}}_M }_{{\bf{E}}^{\rm DAS}} \nonumber  + \underbrace{{\bf{W}}_b  {\bf{H}}_b {\bf{P}} {\bf{P}}^H {\bf{H}}_c^H {\bf{W}}_c^H + {\bf{W}}_c  {\bf{H}}_c {\bf{P}} {\bf{P}}^H {\bf{H}}_b^H {\bf{W}}_b^H \!-\! {\bf{I}}_M }_{{\bf{E}}^{\rm cross}} ,  
\end{align}
\end{small}
\hrulefill
\end{figure*}
where ${\bf{W}}_b\in\mathcal{C}^{M \times N_r}$ and ${\bf{W}}_c \in \mathcal{C}^{M \times a}$ denote the linear equalizers related to the received signals ${\bf{y}}_b$ and ${\bf{y}}_c$, respectively.


As seen from the MSE expression, the average MSE is divided into three sections, i.e., ${{\bf{E}}^{\rm RIS}}$, ${{\bf{E}}^{\rm DAS}}$ and ${{\bf{E}}^{\rm cross}}$. Obviously, ${{\bf{E}}^{\rm RIS}}$ represents the MSE of all data streams induced by $N-a$ passive elements in the reflection mode. 
${{\bf{E}}^{\rm DAS}}$ stands for the MSE generated by an $a$-element distributed antenna array.  
As for the last term ${{\bf{E}}^{\rm cross}}$, it can be regarded as the mutual impact brought by the combination of RIS and distributed antennas at the RDARS-aided system.  
In particular, the system reduces to the conventional fully passive RIS-aided system with $N$ reflecting elements when $a$ equals zero, or ${\bf{A}}={\bf{0}}$. On the other hand,  it becomes identical to the DAS system with $N\!+\! N_r$ distributed antennas when all of the RDARS elements are connected to the BS.

\vspace{-4pt}
\subsection{Problem Formulation}
{\color{black}The average sum MSE minimization problem for the RDARS-aided MIMO system can be formulated as\footnote{\color{black}
The proposed algorithms in this paper can be easily extended to the general case with the inequality constraint $||{\rm diag}({\bf{A}})||_0 \leq a$. 
 Additionally, we conclude that this general problem, subject to the constraint $||{\rm diag}({\bf{A}})||_0 \leq a$, is equivalent to the problem with the equality constraint $||{\rm diag}({\bf{A}})||_0 = a$.
}}
\begin{subequations}\label{Problem P1}
  \begin{align}
  (\text{P1}): \ \min_{{\bf{A}},{\bf{W}},{\bf{\Phi}}} \quad  \!\!&  f_{\rm MSE}^{{\bf{A}},{\bf{W}},{\bf{\Phi}}}={\rm Tr}({{\rm {\bf MSE}}})  \\
 \!\!\mbox{s.t.} \quad 
       & | {\bf{\Phi}} |_{n,n}=1, ~\forall n \in \mathcal{N},  \label{P1_RIS}  \\
       & [{\bf{A}}]_{n,n} \in \{0,1\}, \forall n \in \mathcal{N},  \label{P1_index} \\
       & [{\bf{A}}]_{i,j}=0, \forall i,j \in \mathcal{N}, i \neq j, \label{P1_index_diag} \\
       & ||{\rm diag}({\bf{A}})||_0=a. \label{P1_index_number}
  \end{align} 
 \end{subequations}
 where \eqref{P1_RIS} denotes the unit-modulus constraints for the passive reflection coefficients at the RADAS, and \eqref{P1_index}\eqref{P1_index_diag} imply the mode selection matrix ${\bf{A}}$ is diagonal and has a diagonal value of 0 or 1, indicating whether the RDARS element is connected to the BS.
 The cardinality constraint \eqref{P1_index_number} limits the number of elements in the connection mode to $a$. 
 Unfortunately, problem (P1) is essentially a mixed-integer optimization problem, which is intractable and NP-hard.
 Due to the cardinality and binary constraints, problem (P1) is more challenging than the conventional fully passive RIS-aided system design.
 To the best of our knowledge, there have been very few studies considering dynamically configuring the RIS elements to be active or not. Additionally, existing algorithms are also inapplicable to our works. 
 In the following section, we propose two different algorithms to overcome the challenges arising from unit-modulus, binary, and cardinality constraints.

\section{Proposed Inexact BCD-PDD Algorithm}
In this section, we propose an efficient inexact BCD-based PDD algorithm to solve the non-convex optimization problem (P1) with the guaranteed convergence. 

{\color{black}
Since the objective function $f_{\rm MSE}^{{\bf{A}},{\bf{W}},{\bf{\Phi}}}$ is convex w.r.t. the receive beamforming matrix $\bf{W}$. The optimal solution can be obtained with closed form by taking the first derivative to 0, i.e., 
 \begin{equation} \label{W_opt}
 {\bf{W}}^{\star}={\bf{P}}^H{\bf{ H}}^{H} \left({\bf{ H}} {\bf{P}}{\bf{P}}^H {\bf{ H}}^{H} \!+\!{\bm{\Lambda}} \right)^{-1}.
 \end{equation}
Substituting the optimal receive beamforming matrix $\bf{W}^{\star}$ into the objective function in problem (P1), the newly formulated objective function with two optimization variables is derived as
\begin{align}\label{obj_mse}
  f_{\rm MSE}^{{\bf{A}},{\bf{\Phi}}}
  =&   {\rm Tr} \left\{ ({\bf{I}}_M  \!+\! \sigma_b^{-2}{\bf{R}}_{b}({\bf{A}},{\bf{\Phi}}) \!+\! \sigma_c^{-2}{\bf{R}}_{c}({\bf{A}})   )^{-1} \right\},  
\end{align} 
where $ {\bf{R}}_{b}({\bf{A}},{\bf{\Phi}})={{\bf{P}}^H {\bf{ H}}_b^{H} {\bf{ H}}_b {\bf{P}}}$ and $ {\bf{R}}_{c}({\bf{A}})={{\bf{P}}^H {\bf{ H}}_c^H {\bf{ H}}_c  {\bf{P}} } $    
denotes the received signal covariance matrices related to the effective channels ${\bf{ H}}_b$ and ${\bf{ H}}_c$, respectively. 
}

The sum-MSE minimization problem for jointly optimizing the mode selection matrix and the RDARS's phase shifts is then formulated as:
\begin{subequations}\label{P2}
  \begin{align}
  (\text{P2}): \ \min_{{\bf{A}},{\bf{\Phi}}} \quad  \!\!&  f_{\rm MSE}^{{\bf{A}},{\bf{\Phi}}}  \\
 \!\!\mbox{s.t.} \quad 
       & \eqref{P1_RIS},\eqref{P1_index},\eqref{P1_index_diag},    {\rm Tr} \{  {\bf{A}}   \} = a.  \label{P2_C1}
  \end{align} 
 \end{subequations}
The constraints are generally intractable. 
Moreover, the objective function in problem (P2) encompasses the matrix inversion of a matrix-valued quadratic function, which renders the problem intractable and NP-hard with regard to the binary matrix variable ${\bf{A}}$.
Even when the RDARS's phase shifts are fixed, the subproblem of the mode selection remains a challenging combinatorial optimization problem. While the exhaustive search method can provide the globally optimal solution, its computational complexity grows exponentially with the number of RDRAS elements.
Consequently, the exhaustive search method becomes infeasible for the large-scale RDARS-aided systems.

In the sequel, we begin by providing an equivalent transformation and introducing an auxiliary variable ${\bf{\bar A}}$ to enhance the tractability of the problem.
Then, a PDD-based iterative algorithm is proposed to address such issues, where the inner iteration updates the primal variables by inexactly solving its corresponding augmented Lagrangian (AL) problem via the BCD method, while the outer iteration updates the dual variables and the penalty parameter.

\subsection{ Equivalent Transformation }
Binary optimization problems are known to be generally NP-hard, posing significant challenges in finding optimal solutions. One common approach to tackle the binary optimization problem is to relax the binary constraints and treat the problem as a continuous optimization one, and then project the continuous solution onto the nearest integer values. However, such relaxation technique often results in a notable loss of performance.
Another type of binary optimization method that yields optimal solutions is the branch-and-bound (BB) method. However, the worst-case computational complexity of the BB method increases exponentially, similar to the exhaustive search method. As the number of RDARS's elements grows, the computational time required to find the optimal solution may become prohibitively high.

To balance the complexity and optimality, we propose a variational reformulation of the binary constraints. The constraints can be equivalently transformed into the $l_2$ box non-separable constraints based on the following lemma. 

\begin{lemma} \label{binary}
Assume ${\bf{x}} \in \mathbb{R}^N,{\bf{v}} \in \mathbb{R}^N$ and define $\chi \triangleq \{ ({\bf{x}},{\bf{v}}) | {\bf{0}} \leq {\bf{x}} \leq {\bf{1}}, ||2{\bf{v}}\!-\!{\bf{1}}||_2^2 \leq N, \left<2{\bf{x}}\!-\!{\bf{1}},2{\bf{v}}\!-\!{\bf{1}}\right>=N, \forall {\bf{v}} \}$. Assume that ${\bf{x}},{\bf{v}} \in \chi$, then we have ${\bf{x}} \in \{0,1\}^N$, and ${\bf{x}}={\bf{v}}$.
\end{lemma}
\begin{proof}
See \cite{yuan2017exact} for detailed proof.
\end{proof}

Based on Lemma 1, we have the following proposition with regard to problem (P2).
\begin{proposition}  \label{P_equivalence}
 Define ${\bf{x}}={\rm diag}({\bf{A}}), {\bf{v}}={\rm diag}({\bf{\bar A}})$ and ${\bm{\theta}}={\rm diag}({\bf{\Phi}})$, problem (P2) is then equivalent to the following problem:
\vspace{-12pt}
\begin{subequations}\label{Problem P_joint}
  \begin{align}
  (\text{P3}): \ \min_{  
 {\bf{0}} \leq {\bf{x}} \leq {\bf{1}} ,{\bf{v}},{\bm{\theta}} } \quad  \!\!& f_{\rm MSE}^{{\bf{x}},{\bf{v}},{\bm{\theta}}}   \\
 \!\!\mbox{s.t.} \quad 
       &  ||2{\bf{v}}-1||_2^2 \leq N,  \label{C_v} \\
       & (2{\bf{x}}\!-\!1)^T(2{\bf{v}}\!-\!1)=N, \\
       &  {\bf{x}}^T {\bf{1}} = a \\
       &  |{\bm{\theta}}|_{n}=1, \forall n. \label{C_phi}
  \end{align} 
\end{subequations}
where $f_{\rm MSE}^{{\bf{x}},{\bf{v}},{\bm{\theta}}} = {\rm Tr} \left\{ ({\bf{I}}_M  \!+\! \sigma_b^{-2} {\bf{R}}_{b}({\bf{A}},{\bf{\Phi}}) \!+\! \sigma_c^{-2}{\bf{\bar R}}_{c}({\bf{\bar A}})   )^{-1}\right\}$ and ${\bf{\bar R}}_{c}({\bf{\bar A}})$ is obtained by replacing ${\bf{\bar A}}$ with ${\bf{A}}$ in ${\bf{R}}_{c}({\bf{A}})$.

\end{proposition} 
\begin{proof}
    {\color{black} By virtue of Lemma 1, it is readily inferred that ${\bf{x}} \in \{0,1\}^N$ and ${\bf{x}}={\bf{v}}$ from the $l_2$ box non-separable constraints, i.e., ${\bf{0}} \leq {\bf{x}} \leq {\bf{1}}, ||2{\bf{v}}\!-\!{\bf{1}}||_2^2 \leq N, \left<2{\bf{x}}\!-\!{\bf{1}},2{\bf{v}}\!-\!{\bf{1}}\right>=N$.
    Then, we can rewrite problem (P2) in an equivalent form as shown in problem (P3). 
    Note that the newly formed constraint $(2{\bf{x}}\!-\!1)^T(2{\bf{v}}\!-\!1)=N$ is referred to as the complementarity constraint and it always holds that $(2{\bf{x}}\!-\!1)^T(2{\bf{v}}\!-\!1)\leq N$ for any feasible ${\bf{x}}$ and ${\bf{v}}$.}
\end{proof}

From Proposition \ref{P_equivalence}, the intractable binary constraints have been rewritten in a tractable form by introducing an extra variable.
Unfortunately, problem (P3) is still difficult to solve due to the non-separable objective function and the equality constraints. 
In the following, we propose an efficient PDD method to address the aforementioned challenges.

 We first convert problem (P3) into its AL form by adding the equality constraints as a penalty term to the objective function, given by
  \begin{align} \label{Problem P_joint}
  (\text{P4}): \ \min_{ {\bf{0}} \leq {\bf{x}} \leq {\bf{1}} ,||2{\bf{v}}-1||_2^2 \leq N,|{\bm{\theta}}|_{n}=1, \forall n } \quad  \!\!&  {\mathcal{L}}({\bm{\theta}},{\bf{x}},{\bf{v}},\lambda,{\bm{\nu}} ), 
  \end{align} 
  \noindent where the augmented Lagrange function $ {\mathcal{L}}({\bm{\theta}},{\bf{x}},{\bf{v}},\lambda,{{\nu}})$ are defined as
 \begin{align}
      & {\mathcal{L}}({\bm{\theta}},{\bf{x}},{\bf{v}},\lambda,{{\nu}})   = f_{\rm MSE}^{{\bf{x}},{\bf{v}},{\bm{\theta}}}   \!+\!  {{\nu}} \left( (2{\bf{x}}\!-\!1)^T(2{\bf{v}}\!-\!1)\!-\!N \right)
      \\
     & \quad \!+\! \lambda ( {\bf{x}}^T{\bf{1}}-a ) \!+\! \frac{1}{2\rho} \left( |{\bf{x}}^T{\bf{1}}-a|_2^2 + |(2{\bf{x}}\!-\!1)^T(2{\bf{v}}\!-\!1)\!-\!N|_2^2 \right) \nonumber
 \end{align}
 where  $\rho$ represents the penalty parameter, $\lambda$ and ${{\nu}}$ denote the dual variables associated with the equality constraints ${\bf{x}}^T {\bf{1}} = a$ and $ (2{\bf{x}}\!-\!1)^T(2{\bf{v}}\!-\!1)\!=\!N $, respectively. 
It is observed that as $\rho$ approaches zero, the penalty term is forced to zero, i.e., the equality constraints are satisfied.

By referring to \cite{Shi2017PDD}, the PDD algorithm consists of double-loop iterations. Starting with the penalized form (P4), we respectively discuss the optimization problems in the inner and outer iterations in the sequel.

\vspace{-12pt}
\subsection{Inner Iteration of PDD}
In the inner iteration, we alternatively solve the AL problem with the BCD-type method by dividing the primal variables into three distinct blocks. Since each subproblem remains non-convex, block successive upper bound minimization (BSUM) is applied to address the AL problem. Consequently, the primal variables are iteratively updated with the inexact BCD method.
To this end, we minimize the objective function of (P4) by sequentially updating ${\bm{\theta}},{\bf{x}},{\bf{\bar x}}$ in the inner loop through the following marginal optimizations.

\subsubsection{ \textbf{${\bf{\theta}}$-subproblem} }
We optimize ${\bm{\theta}}$ while treating the remaining variables as constants. The subproblem is then formulated as
\begin{align}
    {\bm{\theta}}^{k+1} = \arg \min_{ |{\bm{\theta}}|_n=1,\forall n }  {\mathcal{L}}({\bm{\theta}},{\bf{x}}^{k},{\bf{v}}^{k},\lambda^{k},{\bm{\nu}}^{k})
\end{align}
where $k$ stands for the $k$-th iteration. Since the subproblem is nonconvex and encounters unit-modulus constraints, there are generally no oracles for deriving the globally optimal solution. 

Hence, we propose to use the majorization-minimization (MM) method in which a sequence of approximate versions of ${\mathcal{L}}$ is solved successively. 
By denoting ${\bf{C}}={\bf{I}}_M \!+\!{\sigma_c^{-2}}{{\bf{P}}^H {\bf{ H}}_{r}^H {\bf{\bar A}}^{k} {\bf{ H}}_{r} {\bf{P}} }$ and ${\bf{X}}_{{\bm{\theta}}}=({\bf{H}}_{d}\!+\!{\bf{G}}^H ({\bf{I}}\!-\!{\bf{A}}^{k} ) {\bf{\Phi}} {\bf{H}}_{r}) {\bf{P}}$, the objective function for ${\bf{\theta}}$-subproblem can be equivalently formulated as $L_{{\bm{\theta}}}$ 
\begin{align}
    {\mathcal{L}}_{{\bm{\theta}}} & = {\rm Tr} \left\{ ( {\bf{C}} + {\bf{X}}_{{\bm{\theta}}}^H {\bm{\Lambda}}_b^{-1} {\bf{X}}_{{\bm{\theta}}} )^{-1 }\right\}, \nonumber \\
    & \overset{(a)}{=} {\rm Tr} \left\{ {\bf{C}}^{-1} \!-\! {\bf{C}}^{-1}{\bf{X}}_{{\bm{\theta}}}^H {\bf{Q}}_{{\bm{\theta}}}^{-1} {\bf{X}}_{{\bm{\theta}}} {\bf{C}}^{-1}  \right\}.
\end{align}
where (a) holds due to  the matrix identity $({\bf{A}}+{\bf{U}}{\bf{B}}{\bf{V}})^{-1}={\bf{A}}^{-1}-{\bf{A}}^{-1} {\bf{U}} ({\bf{I}} + {\bf{B}} {\bf{V}}{\bf{A}}^{-1} {\bf{U}} )^{-1} {\bf{B}} {\bf{V}}{\bf{A}}^{-1}$ in \cite{ zhang2017matrix}  and ${\bf{Q}}_{\bm{\theta}}={\bm{\Lambda}}_b \!+\! {\bf{X}}_{{\bm{\theta}}}{\bf{C}}^{-1} {\bf{X}}_{{\bm{\theta}}}^H$.
Based on the above manipulations, the problem for minimizing $L_{{\bm{\theta}}}$ is equivalent to maximize ${\rm Tr}({\bf{C}}^{-1}{\bf{X}}_{{\bm{\theta}}}^H {\bf{Q}}_{\bm{\theta}}^{-1} {\bf{X}}_{{\bm{\theta}}} {\bf{C}}^{-1})$, which has the matrix projection form w.r.t. ${\bf{X}}_{{\bm{\theta}}}$. Armed with the MM technique, a tight surrogate function can be found as 
\begin{align} \label{upper-bound}
    &{\rm Tr}({\bf{C}}^{-1}{\bf{X}}_{{\bm{\theta}}}^H {\bf{Q}}_{\bm{\theta}}^{-1} {\bf{X}}_{{\bm{\theta}}} {\bf{C}}^{-1}) \nonumber \\
    &\geq  2\Re\{ {\rm Tr}({\bf{C}}^{-1}({\bf{X}}_{{\bm{\theta}}}^{k})^H ({\bf{Q}}_{\bm{\theta}}^{k})^{-1} {\bf{X}}_{{\bm{\theta}}} {\bf{C}}^{-1} ) \} \nonumber \\
    & ~~ -{\rm Tr} ( ({\bf{Q}}_{\bm{\theta}}^{k})^{-1}   {\bf{X}}_{{\bm{\theta}}}^{k} {\bf{C}}^{-2} ({\bf{X}}_{{\bm{\theta}}}^{k})^H ({\bf{Q}}_{\bm{\theta}}^{k})^{-1} {\bf{Q}}_{\bm{\theta}} )
\end{align}
with equality achieved at ${\bf{X}}_{{\bm{\theta}}}={\bf{X}}_{{\bm{\theta}}}^{k}$, 
where ${\bf{X}}_{{\bm{\theta}}}^{k}$ is obtained by substituting $\bm{\theta}$ for $\bm{\theta}^{k}$ in ${\bf{X}}_{{\bm{\theta}}}$. And ${\bf{Q}}_{\bm{\theta}}^{k}$ is similarly defined.

Substituting  ${\bf{X}}_{{\bm{\theta}}}$ into the surrogate function and dropping the irrelevant term, the surrogate function can be simplified as
\begin{align} \label{L_Q}
    {\tilde {\mathcal{L}}}_{{\bm{\theta}}}=2\Re \{ {\rm Tr} (    {\bf{D}} {\rm diag}({\bm{\theta}}) ) \} \!-\! {\rm Tr}({\bf{U}}{\rm diag}({\bm{\theta}}) {\bf{V}} {\rm diag}({\bm{\theta}})^H)
\end{align}
where ${\bf{D}}$ is defined at the top of this page in \eqref{D}, ${\bf{U}}$ and ${\bf{V}}$ are defined as follows 
\begin{align}  
    & {\bf{U}} = ({\bf{I}}\!-\!{\bf{A}}^k) {\bf{G}} ({\bf{Q}}_{\bm{\theta}}^{k})^{-1} {\bf{X}}_{{\bm{\theta}}}^{k} {\bf{C}}^{-2} ({\bf{X}}_{{\bm{\theta}}}^{k})^H ({\bf{Q}}_{\bm{\theta}}^{k})^{-1} {\bf{G}}^H ({\bf{I}}\!-\!{\bf{A}}^k), \nonumber \\
    & {\bf{V}} = {\bf{H}}_{r} {\bf{P}} {\bf{C}}^{-1} {\bf{P}}^H {\bf{H}}_{r}^H . 
\end{align}
\begin{figure*} 
\vspace{-12pt}
\begin{align} \label{D}
    & {\bf{D}} = {\bf{H}}_{r} {\bf{P}} \left( {\bf{I}} \!-\! {\bf{C}}^{-1}{\bf{P}}^H {\bf{H}}_{d}^H ({\bf{Q}}_{\bm{\theta}}^{k})^{-1} {\bf{X}}_{{\bm{\theta}}}^{k} \right) {\bf{C}}^{-2} ({\bf{X}}_{{\bm{\theta}}}^{k})^H ({\bf{Q}}_{\bm{\theta}}^{k})^{-1} {\bf{G}}^H \left({\bf{I}}\!-\!{\bf{A}}^k \right) 
\end{align}
\vspace{-4mm}
\hrulefill
\end{figure*}
With the matrix identity ${\rm Tr}({\bf{A}}{\rm {diag}}({\bf{b}}){\bf{C}}{\rm {diag}}({\bf{b}})^H )={\bf{b}}^H ({\bf{C}}^T \odot {\bf{A}}){\bf{b}}$ \cite{Wang2023HWI}, ${\tilde L}_{{\bm{\theta}}}$ in \eqref{L_Q} can be further transformed into the quadratic form w.r.t. ${\bm{\theta}}$  as
  \begin{align}
  & {\tilde {\mathcal{L}}}_{{\bm{\theta}}} = {\bm{\theta}}^H ({\bf{V}}^T \!\odot \!  {\bf{U}} ){\bm{\theta}} \!-\! 2 \Re\{ {\bm{\theta}}^H {\rm diag}({\bf{D}}^H )  \}
  \end{align} 
Despite the objective ${\tilde L}_{{\bm{\theta}}}$ is convex with regard to ${\bm{\theta}}$, the unit-modulus constraints still render the optimization difficult to solve. 
We apply the MM technique again to obtain a closed-form solution by solving the tight upper-bound problem.
Since  the gradient of ${\tilde L}_{{\bm{\theta}}}$  is easily found to be Lipschitz continuous, we have
\begin{align}
    {\tilde {\mathcal{L}}}_{{\bm{\theta}}} \leq &  {\bm{\theta}}^H {\bf{L}} {\bm{\theta}} + 2\Re\{{\bm{\theta}}^H({\bf{V}}^T \!\odot \!  {\bf{U}}\!-\!{\bf{L}}){\bm{\theta}}^k\} \nonumber \\
    & + ({\bm{\theta}}^k)^H ({\bf{L}}\!-\!{\bf{V}}^T \!\odot \!  {\bf{U}}){\bm{\theta}}^k \!-\! 2 \Re\{ {\bm{\theta}}^H {\rm diag}({\bf{D}}^H )  \}
\end{align}
where ${\bf{L}}$ denotes the Lipschitz constant associated with  ${\tilde {\mathcal{L}}}_{{\bm{\theta}}}$. 
To avoid the high computational complexity induced by the eigenvalue decomposition of ${\bf{V}}^T \!\odot \!  {\bf{U}}$, we select ${\bf{L}}={\rm Tr}( {\bf{V}}^T \odot {\bf{U}} ) {\bf{I}}$ as the Lipschitz constant. 
By virtue of the unit-modulus property of ${\bm{\theta}}$, the quadratic term of the upper bound is equivalently transformed to a constant term,  i.e., ${\bm{\theta}}^H {\bf{L}} {\bm{\theta}} = N {\rm Tr}( {\bf{V}}^T \odot {\bf{U}} )$. 
Based on the above fact, the $\bm{\theta}$-subproblem is finally simplified as 
  \begin{align}
  (\text{P5}): \ \min_{ |{\bm{\theta}}|_{n}=1, \forall n } \quad  \!\!&  \Re\{{\bm{\theta}}^H {\bm{\tilde \theta}}^{k} \}
  \end{align} 
where $ {\bm{\tilde \theta}}^{k} =   ({\bf{V}}^T \odot {\bf{U}} \!-\!  {\rm Tr}( {\bf{V}}^T \odot {\bf{U}} ) {\bf{I}}  ){\bm{\theta}}^{k} \!-\! {\rm diag}({\bf{D}}^H ) $.
The optimal closed-form solution of ${\bm{\theta}}$ at $k$-th iteration is obtained as 
\begin{align} \label{ADM_A}
    { {{\theta}}^{\star}_n = } \left\{
      \begin{aligned}
     - e^{j \angle{{{\tilde \theta}}^{k}_n } }, \qquad \qquad & \quad  {{\tilde \theta}}^{k}_n \neq 0, \\
     {\rm any}~{{\theta}}_n~{\rm with}~|{{\theta}}_n|=1,  & \quad {\rm otherwise},
      \end{aligned}
    \right. \forall n\in \mathcal{N}.
\end{align}

\subsubsection{ \textbf{${\bf{x}}$-subproblem} }
The variable ${\bf{x}}$ is then updated by solving the following box-constrained subproblem
\begin{align}
    {\bf{x}}^{k+1} = \arg \min_{ {\bf{0}} \leq {\bf{x}} \leq {\bf{1}} } {\mathcal{L}}({\bm{\theta}}^{k+1},{\bf{x}},{\bf{v}}^{k},\lambda^{k},{{\nu}}^{k} ).
\end{align}
Due to the inversion of the matrix-valued quadratic function in the objective function, this problem is non-tractable and has no closed-form solution. Similarly, the MM technique is utilized to find a more tractable reformulation of the objective function. Specifically, the tight upper bound of ${\mathcal{L}}({\bm{\theta}}^{k+1},{\bf{x}},{\bf{v}}^{k},\lambda^{k},{{\nu}}^{k} )$ at the point ${\bf{x}}^k$ is derived as ${{\mathcal{L}}}_{\bf{x}}$ 
\begin{align}
    {{\mathcal{L}}}_{\bf{x}} =&  {\rm Tr} ({\bf{C}}^{-1}) \!-\! 2\Re \{ {\rm Tr} ( {\bf{C}}^{-2}({\bf{Y}}_{{\bf{x}}}^{k})^H ({\bf{Q}}_{{\bf{x}}}^{k})^{-1} {\bf{Y}}_{{\bf{x}}}  ) \}  \nonumber \\
    & +\! {\rm Tr} \left( ({\bf{Q}}_{{\bf{x}}}^{k})^{-1}   {\bf{Y}}_{{\bf{x}}}^{k} {\bf{C}}^{-2} ({\bf{Y}}_{{\bf{x}}}^{k})^H ({\bf{Q}}_{{\bf{x}}}^{k})^{-1} {\bf{Q}}_{{\bf{x}}} \right) \nonumber \\
    & \quad \!+\! \lambda ( {\bf{x}}^T{\bf{1}}-a )  \!+\! {{\nu}} \left( (2{\bf{x}}\!-\!1)^T(2{\bf{v}}\!-\!1)\!-\!N \right)
     \nonumber \\
     & \quad \!+\! \frac{1}{2\rho} \left( |{\bf{x}}^T{\bf{1}}-a|_2^2 + |(2{\bf{x}}\!-\!1)^T(2{\bf{v}}\!-\!1)\!-\!N|_2^2 \right)
\end{align}
where $ {\bf{Y}}_{{\bf{x}}}=({\bf{H}}_{d}+{\bf{G}}^H ({\bf{I}}\!-\!{\bf{A}}) {\bf{\Phi}}^{k+1} {\bf{H}}_{r}) {\bf{P}} $ and ${\bf{Q}}_{{\bf{x}}}={\bm{\Lambda}}_b \!+\! {\bf{Y}}_{{\bf{x}}}{\bf{C}}^{-1} {\bf{Y}}_{{\bf{x}}}^H$. ${\bf{Y}}_{{\bf{x}}}^{k}$ is obtained by substituting ${{\bf{x}}}^{k}$ for ${{\bf{x}}}$ in ${\bf{Y}}_{{\bf{x}}}$. 

With some mathematical manipulations and dropping the constant term, the problem for minimizing ${{\mathcal{L}}}_{\bf{x}}$ can be equivalently reformulated as a quadratic minimization problem with the box constraints:
\begin{align}
   (\text{P6}): \ \min_{{\bf{x}}} \quad &{\tilde {\mathcal{L}}}_{\bf{x}}={\bf{x}}^T {\bf{\Xi}} {\bf{x}} +{\bf{x}}^T {\bm{\zeta}} \nonumber \\
    \mbox{s.t.} \quad & {\bf{0}} \leq {\bf{x}} \leq {\bf{1}}. \label{C_01}
\end{align} 
where ${\bf{\Xi}}$ and ${\bm{\zeta}}$ are defined in \eqref{x_sub} at the top of next page.
\begin{figure*}   
\vspace{-12pt}
\begin{align} \label{x_sub}
    
    {\bf{\Xi}} & = \left( {\bm{\Phi}}^{k+1} {\bf{H}}_{r} {\bf{P}} {\bf{C}}^{-1} {\bf{P}}^H {\bf{H}}_{r}^H ({\bm{\Phi}}^{k+1})^H \right)^T \odot \left( {\bf{G}} {\bf{Q}}_{\bf{x}}^{-1} {\bf{Y}}_{{\bf{x}}}^{k} {\bf{C}}^{-2} ({\bf{Y}}_{{\bf{x}}}^{k})^H {\bf{Q}}_{\bf{x}}^{-1} {\bf{G}}^H \right) + \frac{1}{2\rho} \left( {\bf{1}}{\bf{1}}^T \!+\! 4 (2{\bf{v}}-{\bf{1}})(2{\bf{v}}-{\bf{1}})^T  \right) \nonumber \\
    {\bm{\zeta}} & = 2 \Re \left\{ {\rm diag} \left( {\bm{\Phi}}^{k+1} {\bf{H}}_{r} {\bf{P}} ( {\bf{I}} \!-\! {\bf{C}}^{-1}{\bf{P}}^H {({\bf{H}}_{d}\!+\!{\bf{G}}^{H}{\bm{\Phi}}^{k+1} {\bf{H}}_{r} )^{H}} {\bf{Q}}_{\bf{x}}^{-1} {\bf{Y}}_{{\bf{x}}}^{k} ) {\bf{C}}^{-2} ({\bf{Y}}_{{\bf{x}}}^{k})^H {\bf{Q}}_{\bf{x}}^{-1} {\bf{G}}^H  \right) \right\} \nonumber \\
    & \qquad \!+\! \frac{1}{\rho} \left( 
   (\rho \lambda-a ){\bf{1}} \!+\!  2(\rho {{\nu}} \!-\!N )( 2{\bf{v}} \!-\!{\bf{1}} ) -\! 2(2{\bf{v}} \!-\!{\bf{1}})^T{\bf{1}}  ( 2{\bf{v}} \!-\!{\bf{1}} ) 
   \right) 
\end{align}
\vspace{-4mm}
\hrulefill
\end{figure*}
However, due to the indefiniteness of ${\bf{\Xi}}$, the box-constrained optimization problem is generally not a standard convex problem. 
To efficiently address the challenges, we divide the objective into the convex and concave parts and transform it into a box-constrained difference-of-convex (DC) problem. 
The convex-concave procedure (CCP) is then proposed to find a local optimum by replacing the concave terms with a convex upper bound. 

We first divide the indefinite term ${\bf{\Xi}}$ into a positive semi-definite matrix ${\bf{\Xi}}_{p}$ and a negative semi-definite matrix ${\bf{\Xi}}_{n}$ by leveraging the eigenvalue decomposition (EVD). Specifically, we have
\begin{align}
    {\bf{\Xi}}_{p} = {\bf{U}} \frac{{\bm{\Lambda}}+|{\bm{\Lambda}}|}{2} {\bf{U}}^H, \quad 
    {\bf{\Xi}}_{n} = {\bf{U}} \frac{{\bm{\Lambda}}-|{\bm{\Lambda}}|}{2} {\bf{U}}^H,
\end{align}
where ${\bf{\Xi}}={\bf{U}}{\bm{\Lambda}}{\bf{U}}^H$ and $|{\bm{\Lambda}}|$ denotes the element-wise absolute function. 
By approximating the concave part ${\bf{x}}^T {\bf{\Xi}}_{n} {\bf{x}}$ with an upper bound at the current point ${\bf{x}}^{k}$, the problem (P6) is reformulated as 
\begin{align}
    (\text{P7}): \ \min_{{\bf{x}}} \quad &{\bf{x}}^T {\bf{\Xi}}_{p} {\bf{x}} +{\bf{x}}^T {\bm{\zeta}} + {\bf{x}}^T {\bf{\Xi}}_{n} {\bf{x}}^{k} + ({\bf{x}}^{k})^T {\bf{\Xi}}_{n} {\bf{x}} \nonumber \\
    \mbox{s.t.} \quad & \eqref{C_01}.
\end{align}
It is readily inferred that the box-constrained optimization problem is a standard convex problem. Therefore, it can be efficiently solved via the standard numerical solvers, like CVX\cite{cvx}.

\subsubsection{ \textbf{${\bf{v}}$-subproblem }}
We then optimize ${\bf{v}}$ with the fixed $ \{{\bm{\theta}}^{k+1},{\bf{x}}^{k+1} \}$ via solving the following subproblem
\begin{align}
    {\bf{v}}^{k+1} = \arg \min_{ ||2{\bf{v}}-{\bf{1}}||\leq \sqrt{N}  } {\mathcal{L}}({\bm{\theta}}^{k+1},{\bf{x}}^{k+1},{\bf{v}},\lambda^{k},{{\nu}}^{k} ).
\end{align}
With equivalent transformation, the problem turns into 
\begin{align}
    (\text{{\color{black}P8}}): \ \min_{{\bf{v}}} \quad &{{\mathcal{L}}}_{\bf{v}} 
    \quad \mbox{s.t.} \quad ||2{\bf{v}}-{\bf{1}}||\leq \sqrt{N}. 
\end{align}
where 
\begin{align}
    {{\mathcal{L}}}_{\bf{v}} = &  {\rm Tr} \left\{ ({\bf{I}}_M  \!+\! \frac{{\bf{R}}_{b}({\bf{A}},{\bf{\Phi}})}{{\sigma_b^{2}}} \!+\!  \frac{{\bf{P}}^H {\bf{ H}}_{r}^H {\rm diag}({\bf{v}}) {\bf{ H}}_{r} {\bf{P}}}{{ \sigma_c^{-2}}}    )^{-1}\right\} \nonumber \\
   & + \frac{1}{2\rho}  \left|(2{\bf{x}}\!-\!{\bf{1}})^T(2{\bf{v}}\!-\!1)\!-\!N + \rho{\nu} \right|_2^2 
\end{align}
Recall that ${\rm Tr}\{{\bf{X}}^{-1}\}$ is convex on $\mathbb{S}_{++}^{n}$ \cite{tenenbaum2011mse}. The first term of ${{\mathcal{L}}}_{\bf{v}}$ exhibits convexity in ${\bf{v}}$, as convexity is preserved through composition with the affine mapping of  ${\bf{v}}$. 
Combining with the quadratic term in ${\mathcal{L}}_{\bf{v}}$, problem (P7) is indeed a convex optimization problem.
Despite its convexity, an optimal closed-form solution for problem (P7) is not directly available.
Nevertheless, the optimal solution can be efficiently obtained by employing a standard CVX solver.

\subsection{ Outer Iteration of PDD }
For the outer iteration of the PDD procedure, the dual variables ${\lambda^{k}},{\bm{\nu}}^{k}$ and the penalty factor $\rho$ are updated. These variables are updated by optimizing the constraint violation quantity, defined as:
{\color{black}
\begin{align}
    h = \max \left( |{\bf{x}}^T{\bf{1}}-a|,  |(2{\bf{x}}\!-\!1)^T(2{\bf{v}}\!-\!1)\!-\!N|  \right)
\end{align}}
The outcome of this optimization yields the update scheme as follows: if $h < \epsilon$, we update the dual variables via the recursive equalities:
 \begin{align}
     &\lambda^{k+1} = \lambda^{k} + \frac{1}{\rho^{k}}({\bf{x}}^T{\bf{1}}-a), \\
     & {\bm{\nu}}^{k+1} = {\bm{\nu}}^{k} + \frac{1}{\rho^{k}}( (2{\bf{x}}\!-\!1)^T(2{\bf{v}}\!-\!1)-N ).
 \end{align}
 Otherwise, we decrease the penalty parameter $\rho$ to increase the penalty, i.e.,
 \begin{align}
     \rho^{k+1} = {\alpha} \rho^{k} 
 \end{align}
where ${\alpha}\in (0,1)$.
Therefore, the PDD method adaptively switches between the AL and the penalty method. This adaptive strategy is expected to find an appropriate penalty parameter $\rho$, where the PDD method could eventually converge.

\subsection{Discussion}
Finally, the inexact BCD-based PDD algorithm for solving the problem (P2) is summarized in {\textbf{Algorithm~\ref{inexact BCD-based PDD}}}. 
The proposed algorithm is guaranteed to converge to a stationary solution via an iterative update of the primal and dual variables, as well as the penalty factor \cite{shi2020penalty}.
The initial values of the dual variables and the penalty parameter should be selected carefully to ensure the penalty term remains sufficiently small. 
Additionally, instead of setting a maximum number of iterations, we choose the relative objective progress (RBP) condition as the termination criterion, i.e.,
\begin{align}
    \frac{|{\mathcal{L}}({\Upsilon}^{k}) - {\mathcal{L}}({\Upsilon}^{k-1})|}{{\mathcal{L}}({\Upsilon}^{k-1})}\leq \epsilon   
\end{align}
where ${\Upsilon} \triangleq \{ {\bm{\theta}},{\bf{x}},{\bf{v}},\lambda,{{\nu}}\}$.
 The computational complexity of {\textbf{Algorithm~\ref{inexact BCD-based PDD}}} primarily stems from the update of primal variables in steps 3-5. In step 3, the computation of matrix inversion and multiplication contributes to a complexity of $O(N^2M+N_r^3)$. Steps 4 and 5 involve the numerical calculation of the CVX solver, which has a complexity of $O(N^{3.5}\log(1/\epsilon))$, with $\epsilon$ representing the desired accuracy.
 Therefore, the total complexity of {\textbf{Algorithm~\ref{inexact BCD-based PDD}}} can be estimated as $O(I_{PDD}(2N^{3.5}\log(1/\epsilon)+N^2M+N_r^3))$, where $I_{PDD}$ denotes the number of iterations for the proposed inexact BCD-based PDD algorithm.

\subsection{Special Cases} 
{\color{black} In this subsection, we focus primarily on the RDARS-aided single-user system and explore two specific cases, i.e., the single-antenna BS case and the line-of-sight (LoS) case. We aim to investigate the intuitive beamforming and the mode selection solution for the MSE minimization problem in these cases.

In the RDARS-aided one single-antenna user communication system, the average MSE minimization problem (P2) degenerates to the following maximization problem 
\begin{align}
   (\text{P9}): \max_{{\bm{\theta}},{\bf{A}}} &\quad ||{\bf{h}}_{d} \!+\! {\bf{G}}^H ({\bf{I}} \!-\! {\bf{A}}) {\rm diag}({\bm{\theta}}) {\bf{h}}_{r} ||_2^2  +{\bf{h}}_{r}^H {\bf{A}} {\bf{h}}_{r} \nonumber \\
    \mbox{s.t.} &\quad \eqref{P2_C1} 
\end{align} 
where ${\bf{h}}_d$ and ${\bf{h}}_r$ denote the reduced channels from the users to the BS and the RDARS, i.e., ${\bf{H}}_d$ and ${\bf{H}}_r$, in the RDARS-aided single-user scenario, respectively.
Then, the optimal closed-form solutions are derived to give insightful observations for the intuitive beamforming design and the mode selection optimization in these two special cases. 

\subsubsection{Single-antenna BS case} 
Assuming the BS has only one antenna, we have the following insights for the beamforming design and the mode selection.
\begin{proposition}\label{pro_SISO}
     For the RDARS-aided one single-antenna user system with single-antenna BS, the selection matrix ${\bf{A}}$ can be simplified to select the first $a$ largest channels between the user and the RADRS, i.e., 
    $
        {\mathcal{A}=\{ c_1,c_2,...,c_a \} },
    $
    where the indices $c_i$ are defined according to the rank ordering of the channel gain, i.e., $h_{r,c_1}^2 \!\geq\! \dots \!\geq\! h_{r,c_N}^2 $. 
    The reflection coefficients are optimally solved by ${\bm{\theta}}^{\star}=e^{j (\angle{ {{h}}_{d} }- \angle{{\bf{h}}_{r}} + \angle{{\bf{g}}}) }$, where ${\bf{g}}$ denotes the reduced channel ${\bf{G}}$ in the single-antenna BS case.    
\end{proposition}}
\begin{proof}
    See Appendix A for the detailed proof.
\end{proof}

\subsubsection{LoS case}
Assuming the BS and the RDARS are properly deployed, the deterministic LoS links exist in the RDARS-related channels. Then, we conclude the following proposition.
\begin{proposition}\label{pro_LoS}
     {\color{black}
     Assume the LoS BS-RDARS channel ${\bf{G}}=\kappa_g {\bf{a}}_{gr}  {\bf{a}}_{gt}^H$, where ${\bf{a}}_{gt}$ denotes the transmit steering vector of the BS antennas and ${\bf{a}}_{gr}$ represents the receive steering vector of the RDARS elements. Similarly, the LoS RDARS-user channel is assumed to be ${\bf{h}}_r=\kappa_r {\bf{a}}_r$ with the RDARS transmit steering vector ${\bf{a}}_r$. $\kappa_g$ and $\kappa_r$ denote the large-scale channel coefficients for BS-RDARS and RDARS-user channels, respectively.
     Then, the selection matrix ${\bf{A}}$ can be chosen arbitrarily and the reflection coefficients can be optimally solved with 
    $ {\bm{\theta}}^{\star}=e^{j(\angle{\bf{{\tilde a}}_{r}} - \angle{{\bf{h}}_{d}^H  {\bf{ a}}_t})}.
    $, where ${\bf{\tilde a}}_r=\kappa_g{\bf{a}}_{gr}^H ({\bf{I}} \!-\! {\bf{A}}) {\rm diag}({\bf{h}}_{r})$.}
\end{proposition}

\begin{proof}
    See Appendix B for the detailed proof.
\end{proof}
Proposition \ref{pro_LoS} indicates the optimal mode selection is significant to fully unleash the potential of the RDARS unless both BS-RADRS and RDARS-user channels are LoS channels. In other words, an optimized selection strategy is necessary to enhance the performance of the RDARS-aided system in fading channels. 
A similar conclusion can be derived for the multi-user case under LoS channels, though the proof is omitted here for simplicity.

\begin{algorithm} \label{inexact BCD-based PDD}
    \normalsize
      \caption{Inexact BCD-based PDD Algorithm }
     \SetKwInOut{Input}{Input}
     \SetKwInOut{Output}{Output}
     \Input{System parameters  $M,N_r,N,a$, the threshold $\epsilon$, etc.}
     \Output{${\bf{A}}^{\star}$, ${\bf{W}}^{\star}$  and RIS phase shifts ${\bm{\Phi}}^{\star}$.}
     Initialize $ ({\bm{\theta}}^{0},{\bf{x}}^{0},{\bf{v}}^{0})$, and set the initial $\lambda,{\bm{\nu}},{\rho}$.\\ 
      \Repeat{ RBP termination criterion is met }
      {

       update ${\bm{\theta}}$ via (21), \\
       update ${\bf{x}}$ by solving (24) via CVX solver, \\ 
       update ${\bf{v}}$ by solving (30) via CVX solver,\\
       if $h<\epsilon$ \\
       \quad update $\lambda$ via (33) and $\bm{\nu}$ via (34) \\
       else \\
       \quad update $\rho$ via (35) \\
       end \\
       $k = k+1$
      }
      Return ${\bf{W}}^{\star}$ in \eqref{W_opt}, ${\bf{\Phi}}^{\star}={\rm diag}({\bm{\theta}})$, ${\bf{A}}^{\star}={\rm diag}({\bf{x}})$.        
\end{algorithm}

\section{Low-Complexity AO Algorithm}
Since the BCD iterations in the PDD algorithm involve the numerical calculation of the CVX solver, the computational complexity is unaffordable as the number of RDARS's elements increases.
In the following, we first approximate the objective function and subsequently present a greedy-search-based AO  algorithm, significantly reducing the computational complexity while yielding a near-optimal solution.

\vspace{-12pt}
\subsection{Mode Selection Optimization} 
Referring to \eqref{obj_mse}, the mode selection matrix ${\bf{A}}$ has a crucial impact on the effective channel related to the elements in the reflection mode between the BS and users, i.e., ${\bf{H}}_b$, as well as the effective channel related to the elements in the connection mode, ${\bf{H}}_c$.
Rewriting ${\bf{H}}_b$ as ${\bf{H}}_b={\bf{H}}_b'-{\bf{G}}^H {\bf{A}} {\bf{\Phi}} {\bf{H}}_{r}$, where ${\bf{H}}_b'={\bf{H}}_{d}+{\bf{G}}^H {\bf{\Phi}} {\bf{H}}_{r}$,  it is reasonable to ignore the term ${\bf{G}}^H {\bf{A}} {\bf{\Phi}} {\bf{H}}_{r}$  and approximate ${\bf{H}}_b$ as ${\bf{H}}_b'$ due to the multiplicative fading effect and the small number of elements in the connection mode.  
Here, ${\bf{H}}_b'$  represents the effective channel of the RDARS where whole elements work in the reflection mode.
Based on the aforementioned approximation, the objective function $f_{\rm MSE}^{ {\bf{A}},{\bf{\Phi}} }$ can be approximated as $f_{\rm MSE}^{app}$, given by:
\begin{align}\label{obj_mse_app}
    &f_{\rm MSE}^{app} = {\rm Tr} \left\{ ({\bf{I}}_M  \!+\! {\bf{\bar R}}_{b}({\bf{\Phi}}) \!+\! { \sigma_c^{-2}} {\bf{R}}_{c}({\bf{A}})   )^{-1} \right\} ,
\end{align}
where ${\bf{\bar R}}_{b}({\bf{\Phi}})={\bf{P}}^H {\bf{H}}_b^{'H}  {\bm{\Lambda}}_B^{-1} {\bf{H}}_b^{'}{\bf{P}}$. 

Due to ${\bf{R}}_{c}({\bf{A}})= {{\bf{P}}^H {\bf{ H}}_{r}^H {\bf{A}} {\bf{ H}}_{r} {\bf{P}} }$, the mode selection optimization problem turns into selecting a sub-channel matrix of size $a\times M$ leading to the lowest value of $f_{\rm MSE}^{app}$  from the entire channel between users and the RDARS, i.e., ${\bf{ H}}_{r}$. 
Let ${\bf{H}}_{r,x}$ represent the submatrix of $x$ rows chosen from ${\bf{H}}_{r}$ after $x$ selections and $f_{{\rm MSE},x}^{app}$ denotes the approximated MSE value at $x$-th selection, $x=1,2,...,a$. 
We have $f_{{\rm MSE},x}^{app} =   {\rm Tr} \left\{ ( {\bf{I}}_M \!+\! {\bf{\bar R}}_{b}({\bf{\Phi}})  \!+\! {\sigma_c^{-2}} {\bf{P}}^H {\bf{H}}_{r,x}^{H}  {\bf{H}}_{r,x}    {\bf{P}} )^{-1} \right\}$.
Specially, $f_{\rm MSE,0}^{app}\!=\!{\rm Tr} \left\{ \left( {\bf{I}}_M \!+\! {\bf{\bar R}}_{b}({\bf{\Phi}}) \right)^{-1} \right\}$. 
When $j$-th row of ${\bf{H}}_{r}$, i.e., ${\bf{h}}_{j}\in \mathbb{C}^{1\times M}$, is selected from the rest of candidate rows in the $x+1$ step, the newly formed $(x+1) \times M$ submatrix ${\bf{H}}_{r,x+1}$ is represented by $[{\bf{H}}_{r,x}^H,{\bf{h}}_{j}^H]^H$. 
Then, the approximated MSE value in the $x+1$-th selection is derived as
\begin{align}
   f_{{\rm MSE},x+1}^{app} & = {\rm Tr} \left\{ \left( {\bf{I}}_M \!+\!  {\bf{\bar R}}_{b}({\bf{\Phi}})  \!+\! {\sigma_c^{-2}} {\bf{P}}^H {\bf{H}}_{r,x+1}^{H}   {\bf{H}}_{r,x+1}    {\bf{P}} \right)^{-1} \right\} \nonumber \\
  & \overset{(b)}{=} f_{\rm MSE,x}^{app} -  \Delta_{j,x}
\end{align}
where ${\bf{M}}_{x}={\bf{I}}_M \!+\!  {\bf{\bar R}}_{b}({\bf{\Phi}})  \!+\! {\sigma_c^{-2}} {\bf{P}}^H {\bf{H}}_{r,x}^{H}   {\bf{H}}_{r,x}    {\bf{P}}$ and $\Delta_{j,x}$ is defined as 
$
  \Delta_{j,x}=\frac{ {\bf{h}}_{j} {\bf{P}} {\bf{M}}_{x}^{-2} {\bf{P}}^H {\bf{h}}_{j}^H }{ {\sigma_c^2} + {\bf{h}}_{j} {\bf{P}} {\bf{M}}_{x}^{-1} {\bf{P}}^H {\bf{h}}_{j}^H  }  
$.
The equation (b) holds with the Sherman-Morrison identity, i.e.,  $({\bf{B}}+{\bf{u}}{\bf{v}}^T)={\bf{B}}^{-1}-\frac{{\bf{B}}^{-1}{\bf{u}}{\bf{v}}^T {\bf{B}}^{-1}}{1+{\bf{v}}^T{\bf{B}}^{-1}{\bf{u}}}$. 
$\Delta_{j,x}$ represents the MSE reduction when choosing $j$-th row of ${\bf{H}}_{r}$ in the $x+1$-th step and remains a positive value when ${\bf{h}}_{j} \neq {\bf{0}}$. 

Based on the aforementioned analysis, a greedy search method can be utilized to find a high-quality solution \cite{Gharavi2004GS}. Specifically, in the $x+1$-th step, the MSE reduction $\Delta_{j,x}$ of all rows in the candidate sets is computed and compared to identify the best one that contributes the most to the MSE value. Then, the selected row is added to the sub-channel matrix in the next selection step. 
In other words, the submatrix ${\bf{H}}_{r,x+1}$ is updated by $[{\bf{H}}_{r,x}^H,{\bf{h}}_{J_{x}}^H]^H$, where 
\begin{align}\label{J}
    J_{x}= \arg \max_{j} \Delta_{j,x}.
\end{align}
The update rule for the next step is then calculated as:
\begin{align} \label{delta_J}
  \Delta_{j,x\!+\!1}
  & =\frac{ {\bf{h}}_{j} {\bf{P}} {\bf{M}}_{x\!+\!1}^{-2} {\bf{P}}^H {\bf{h}}_{j}^H }{ {\sigma_c^2} + {\bf{h}}_{j} {\bf{P}} {\bf{M}}_{x\!+\!1}^{-1} {\bf{P}}^H {\bf{h}}_{j}^H  },  
\end{align}
where 
$
    {\bf{M}}_{x\!+\!1}^{-1} = {\bf{M}}_{x}^{-1} -
    \frac{  {\bf{M}}_{x}^{-1} {\bf{P}}^H {\bf{h}}_{J_{x}}^H {\bf{h}}_{J_{x}} {\bf{P}}  {\bf{M}}_{x}^{-1}  }{  {\sigma_c^2} + {\bf{h}}_{J_{x}} {\bf{P}} {\bf{M}}_{x}^{-1} {\bf{P}}^H {\bf{h}}_{J_{x}}^H  }.
$.

Finally, the high-quality sub-channel matrix contributing most to the objective function is found after $a$ selection steps, where the selected sub-channel index corresponds to the mode index in the connection mode on RDARS.

\subsection{Beamforming Design}
With the fixed mode selection matrix ${\bf{A}}$, the phase coefficients optimization problem becomes
  \begin{align}
  (\text{P10}): \ \min_{ {\bf{\Phi}} } \quad  \!\!   f_{\rm MSE}^{app}   \qquad \mbox{s.t.}  \eqref{P1_RIS}.
  \end{align} 
Noting the problem (P10) has the same structure as the ${\bm{\theta}}$-subproblem in Section III.B, the proposed MM technique for solving ${\bm{\theta}}$-subproblem can be applied to address the problem (P10) with semi-closed-form solutions.
Due to length limitations, we omit this procedure for simplicity. 

The procedure for the proposed low-complexity AO algorithm is meticulously outlined in {\textbf{Algorithm~\ref{greedy search}}}. 
The computational complexity of the greedy-search-based mode selection is estimated to be $O(aNM^3)$ \cite{Gao2018GS}. On the other hand, the MM-based beamforming optimization incurs a complexity of $O(N^2M+N_r^3)$. Consequently, the aggregate complexity of {\textbf{Algorithm~\ref{greedy search}}} can be approximated as $O(I_{AO}(aNM^3 + N^2M+N_r^3))$, with $I_{AO}$ denoting the number of iterations.
The complexity of the greedy search-based AO algorithm is substantially lower than the inexact BCD-based PDD algorithm, as observed from the complexity analysis.
In conclusion, {\textbf{Algorithm~\ref{greedy search}}} offers a computationally efficient solution to tackle the intractable optimization problem. Such integration of the greedy-search-based mode selection and the MM-based beamforming optimization achieves an acceptable performance while upholding a manageable computational complexity.

\begin{algorithm} \label{greedy search}
    \normalsize
      \caption{Low-Complexity AO Algorithm }
     \SetKwInOut{Input}{Input}
     \SetKwInOut{Output}{Output}
     \Input{System parameters  $M,N_r,N,a$, the threshold $\epsilon$, etc.}
     \Output{${\mathcal{A}}^{\star}$, ${\bf{W}}^{\star}$  and RIS phase shifts ${\bm{\Phi}}^{\star}$.}
     Set the initial point $ ({\bf{A}}^{0},{\bf{W}}^{0},{\bm{\Phi}}^{0})$.\\ 
      \Repeat{ $| f_{\rm MSE}^{app}({\mathcal{A}}_{k},{\bm{\theta}}_{k}) - f_{\rm MSE}^{app}({\mathcal{A}}_{k\!-\!1},{\bm{\theta}}_{k\!-\!1}) | \leq \epsilon $ }
      {

        update ${\bf{M}}_0^{-1}$, $\mathcal{N}=\{1,2,...,N\}$, \\
        \Repeat{$x=1:a$}
        {
        update $J_{x}$ in \eqref{J}, \\
        update ${\mathcal{N}=\mathcal{N}-\{ J_{x} \}}$, \\
        for $j \in \mathcal{N} $ \\
        \quad 
        update ${\bm{\Delta}}_{j,x+1}$ via \eqref{delta_J}, \\
        end \\
        
        } 
       update ${\mathcal{A}}_{k}=\{1,2,...,N\}-\mathcal{N}$, \\ 
       update ${\bm{\theta}}_{k}$ by solving problem (P9),
      }
      Return ${\bf{W}}^{\star}$ in \eqref{W_opt}, ${\bf{\Phi}}^{\star}={\rm diag}({\bm{\theta}}_{k})$, ${\mathcal{A}}^{\star}={\mathcal{A}}_{k}$.      
  \end{algorithm}


\section{Simulations}
In this section, we provide numerical simulations to evaluate the performance of all proposed algorithms for the RDARS-aided communication system. Unless otherwise stated, we set the basic system parameters as follows: $N_r=4$, $M=4$, $N=256$. 
{\color{black} Under a three-dimensional deployment setup as shown in Fig.~\ref{system topologogies}, the coordinate of the BS is set as $(0,200,5)$m and the users are randomly distributed in a circle with a radius of $10$m and a center of $(0,0,1.5)$m.
The location of RDARS is $(30,d_{R},15)$m, where $d_{R}$ denotes the distance of the RDARS along the y-axis.  }
The path loss model at 3GPP \cite{cho2010mimo} is considered, i.e., $\text{PL}(d)\!=\!\beta_0 \Psi_a{(d)}^{-{\alpha}_{p}}$,  where  $d$  and ${\alpha}_{p}$ respectively denote the propagation distance and the path loss exponent. $\beta_0$ stands for the path loss at the reference distance 1m and $\Psi_a$ denotes the random shadowing impact subject to the 
{\color{black} Gaussian distribution with zero mean and a standard derivation of $\tilde \sigma$. We assume $\beta_0\!=\!-30$dB and $\Psi_a \sim \mathcal{N}(0,{\tilde \sigma}^2 )$ with  ${\tilde \sigma}=5.8$dB in this simulation. }
In particular,  we set ${\alpha}_{p}^{RB}=2.2$,  ${\alpha}_{p}^{UR}=2.2$ and ${\alpha}_{p}^{UB}=3.5$ for the BS-RDRAS channel ${\bf H}_{RB}$, the RDRAS-user channels ${\bf H}_{UR}$ and BS-user channels ${\bf H}_{UB}$, respectively.  

{\color{black} As for the small-scale fading, we assume the spatially correlated Rician fading channels, modeled as
\begin{align}
    {\bf{H}} = \sqrt{ \kappa_t }  {\bf{H}}^{\rm LoS} +  \sqrt{ 1-\kappa_t }  {\bf{H}}^{\rm NLoS}, \label{Rician_channel}
\end{align}
where ${\bf{H}}^{\rm LoS}$ and ${\bf{H}}^{\rm NLoS}$ represent the deterministic LoS and spatial non-LoS (NLoS) components, respectively.
$\kappa_t$ in \eqref{Rician_channel} represents the normalized Rician factor belonging to $[0,1]$. With $\kappa_t$ increasing from 0 to 1, the channel ${\bf{H}}$ transitions from a Rayleigh fading channel to an LoS channel.

The deterministic LoS component is determined by the antenna array response vectors, i.e., ${\bf{h}}_{UB,m}^{\rm LoS}={\bf{a}}_B(\psi_{UB,m})$, ${\bf{G}}^{\rm LoS}={\bf{a}}_{R}(\psi_A,\vartheta_A) {\bf{a}}_B(\psi_{RB})^H$ and ${\bf{h}}_{UR,m}^{\rm LoS}={\bf{a}}_R(\psi_{D,m},\vartheta_{D,m})$, where ${\bf{a}}_B$ and ${\bf{a}}_{R}$ denote the array response vectors at the BS and the RDARS, respectively. $\psi_{UB,m}$ and $\psi_{RB}$ are the angles of departure (AoDs) at the BS towards the $m$-th user and the RDARS, respectively. Furthermore, $\psi_A$ and $\vartheta_A$ denote the azimuth and elevation angles of arrival (AoAs) at the RDARS, respectively, whereas  $\psi_D$ and $\vartheta_D$ stand for the AoDs. 

The NLoS component ${\bf{H}}^{\rm NLoS}$ is assumed to be correlated Rayleigh fading, modeled as 
\begin{align}
    {\bf{H}}^{\rm NLoS} = {\bf{R}}^{\frac{1}{2}}{\bf{W}}{\bf{T}}^{\frac{1}{2}}
\end{align}
where ${\bf{W}}$ denotes the random components of each channel with independent and identically distributed (i.i.d.) zero-mean unit covariance complex Gaussian random elements.
${\bf{R}}$ and ${\bf{T}}$ are deterministic nonnegative semi-definite receive and transmit spatial correlation matrices, respectively. Specifically, ${\bf{H}}^{\rm NLoS}$ becomes the conventional Rayleigh fading channel when ${\bf{R}}={\bf{T}}={\bf{I}}$.
We adopt the exponential correlation model to characterize the spatial correlation among the antennas \cite{R1,R2}. For the ULA topology at the BS, the (i,j)-th element of correlation matrix $R_{B}$ is given as $R_{B}(i,j)=\beta_{R_{B}}^{|i-j|}$. Regarding the spatial transmit and receive correlation matrices ${T}_{R}$ and ${R}_{R}$ at the RDARS, we apply the Kronecker product of the vertical correlation matrix ${R}_{v}$ and the horizontal correlation matrix ${R}_{h}$ to the correlation model and assume ${T}_{R}={R}_{R}$, i.e., ${T}_{R}={R}_{R}={R}_{v}\otimes {R}_{h}$ \cite{R3,R4}. Using the exponential correlation model, the element in ${R}_{v}$ and ${R}_{h}$ are given as $R_{v}(i,j)=\beta_{R_{v}}^{|i-j|}$ and $R_{h}(i,j)=\beta_{R_{h}}^{|i-j|}$, respectively. $\{ \beta_{R_{B}}, \beta_{R_{v}}, \beta_{R_{h}} \}$ denotes the correlation coefficients between any two adjacent antennas or elements at the BS and the RDARS. 
The correlation coefficients at the RDARS are set as $\beta_{R_{v}}=\beta_{R_{h}}=\beta_{c}=0.5$. In contrast, the correlation coefficient at the BS is set as $\beta_{R_{B}}=0$ unless otherwise stated.
}

{\color{black} Specifically, we assume the channel between users and the BS ${\bf{H}}_{d}$ is assumed to be Rayleigh fading with Rician factor $\kappa_t$ being 0, while the RDARS-related channels ${\bf{H}}_{r}$ and ${\bf{G}}$ obey Rician fading with Rician factor $\kappa_t$ being 0.75 in this simulation unless otherwise stated.
The noise powers $\sigma_b^2$ and $\sigma_c^2$ are given by $\sigma_b^2=\sigma_c^2=\sigma^2=-90$dBm.
The average normalized MSE (ANMSE) denoted by  $ {\rm ANMSE}\triangleq {\rm Tr}({\bf MSE})/M$  is adopted as the performance metric.  
All simulation results are obtained by averaging over $300$ channel realizations. }

    \begin{figure}[t]
        \centering  
        \includegraphics[width=0.45\textwidth]{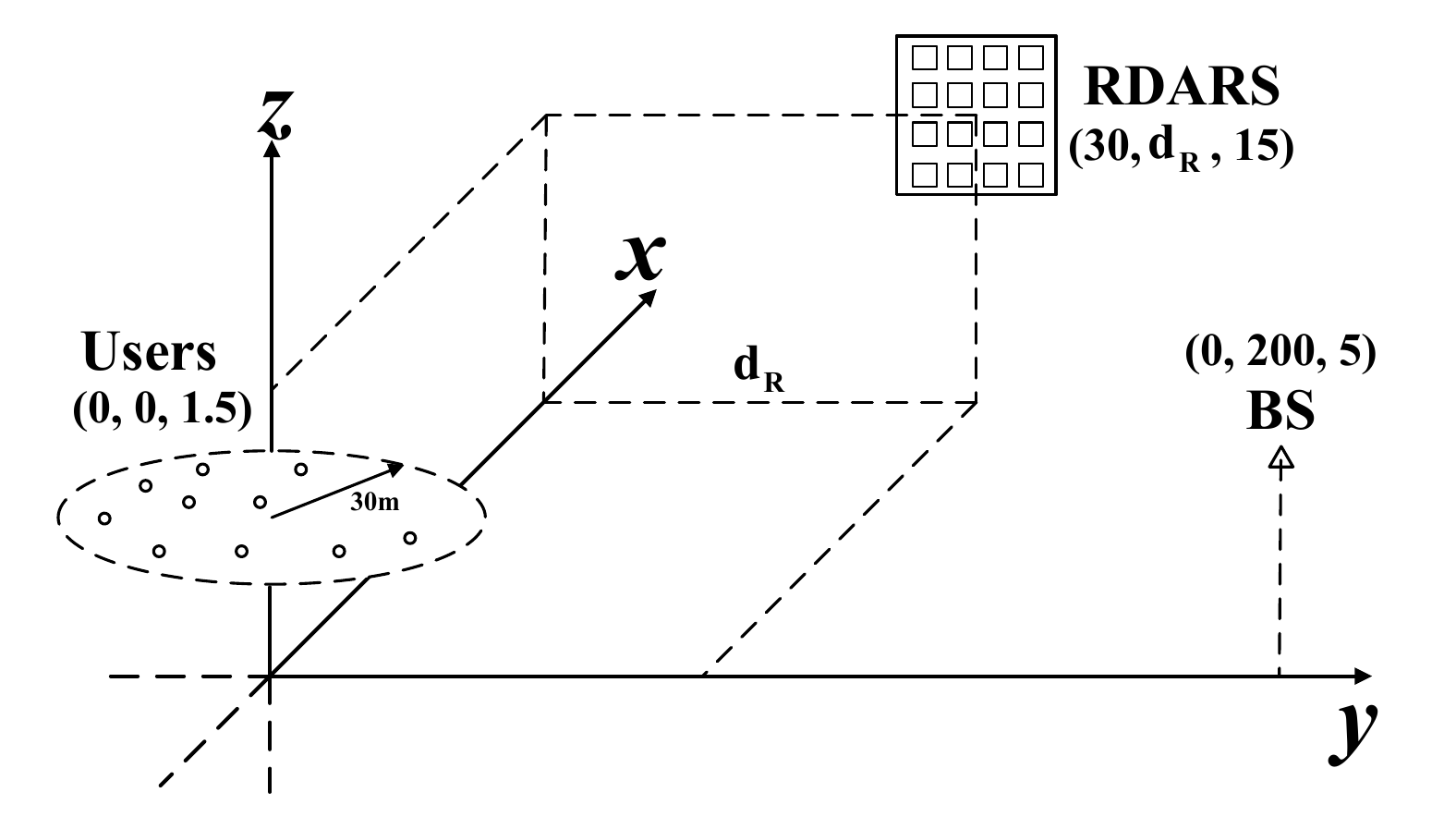}
        \caption{Three-dimensional coordinates of the system deployment for simulations.}
        \label{system topologogies}
    \end{figure}

For the purpose of comparison, the following schemes are utilized in the simulation:
1) {\textbf{Passive RIS}}: This scheme represents the conventional fully passive RIS, serving as the benchmark. The number of passive elements in the RIS is the same as that in the RDARS system.
2) {\textbf{DAS}}: Similar to the passive RIS, the DAS system is employed as a comparative scheme, where the number of distributed antennas matches the number of connected elements in the RDARS system.
3) {\textbf{Random Index}}:  In this scheme, the indices of elements in the connection mode on RDARS are randomly selected, with optimized phase coefficients assigned to the passive elements.
4) {\textbf{Fixed Index}}: Similarly, in the ``Fixed Index" scheme, the first $a$ elements are designated in the connection mode.
5) {\textbf{GS-AO}}: The ``GS-AO" method refers to the proposed algorithm outlined in Section IV, which alternately optimizes the mode selection matrix using the greedy search approach and the phase coefficients using the MM technique.
6) {\textbf{IBCD-PDD}}: The proposed inexact BCD-based PDD algorithm is referred to as ``IBCD-PDD".

\subsection{Convergence.}
The convergence behaviors of the proposed PDD algorithm under various initialization schemes are depicted in Fig.~\ref{PDD_converve}.
The initial mode selection index is chosen as the ``Random Index" in initialization 1 and the ``Fixed Index" in initialization 2. Both initialization schemes adopt random phase coefficients. 
Fig.~\ref{PDD_converve} (a) illustrates the constraints violation during the PDD iteration.  
As the number of iterations increases, the penalty is adaptively adjusted to gradually force the constraint violation to approach the predefined accuracy.
Regardless of the initial conditions or parameter settings, the PDD algorithm consistently guarantees feasibility.
From Fig.~\ref{PDD_converve} (b), the PDD algorithm exhibits monotonically convergent behavior, with the ANMSE value converging towards a maximum value. This can be attributed to the relatively small initial penalty. Thus, the solution obtained in the first few iterations fails to satisfy the equality constraints, ultimately leading to the incremental behavior. 
Different initialization schemes can significantly influence the effectiveness of convergence, as an appropriate initial point plays a pivotal role in achieving optimal outcomes.

    \begin{figure}[t]
        \centering  
        \subfigure[$a$]{
        \includegraphics[width=0.24\textwidth]{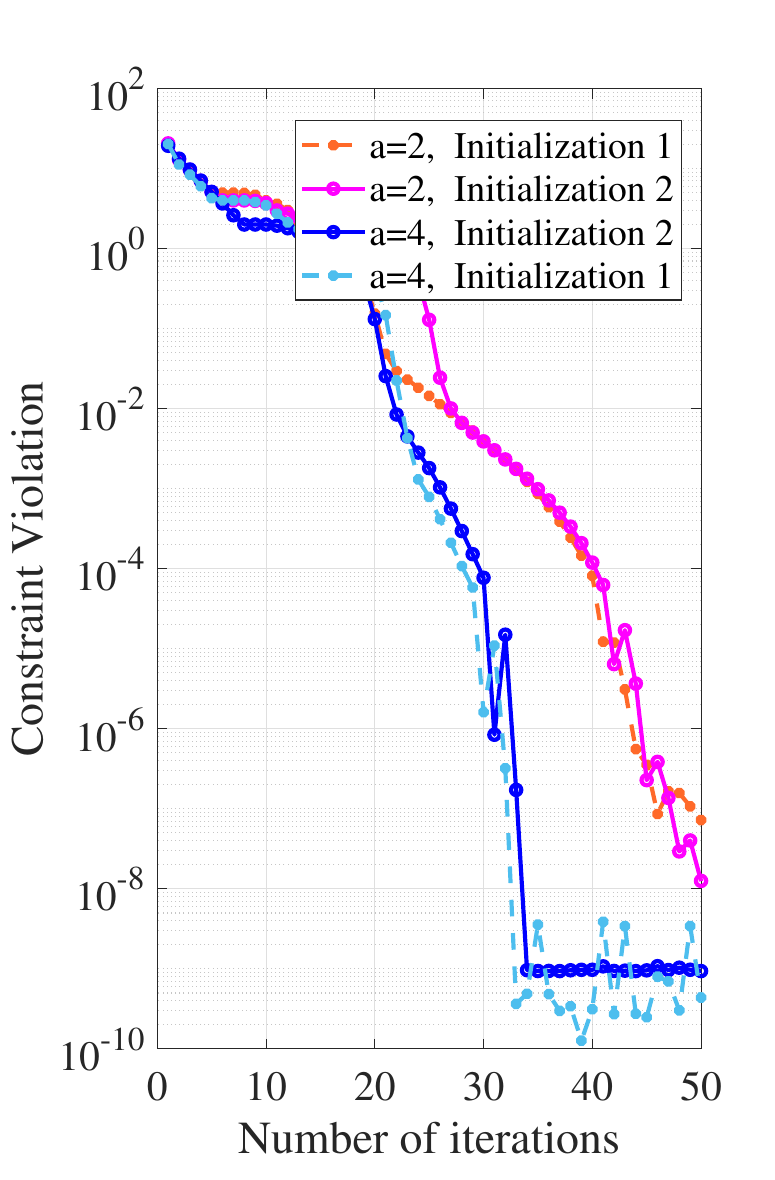}}
        \hspace{-6.5mm}
        \subfigure[$b$]{
        \includegraphics[width=0.24\textwidth]{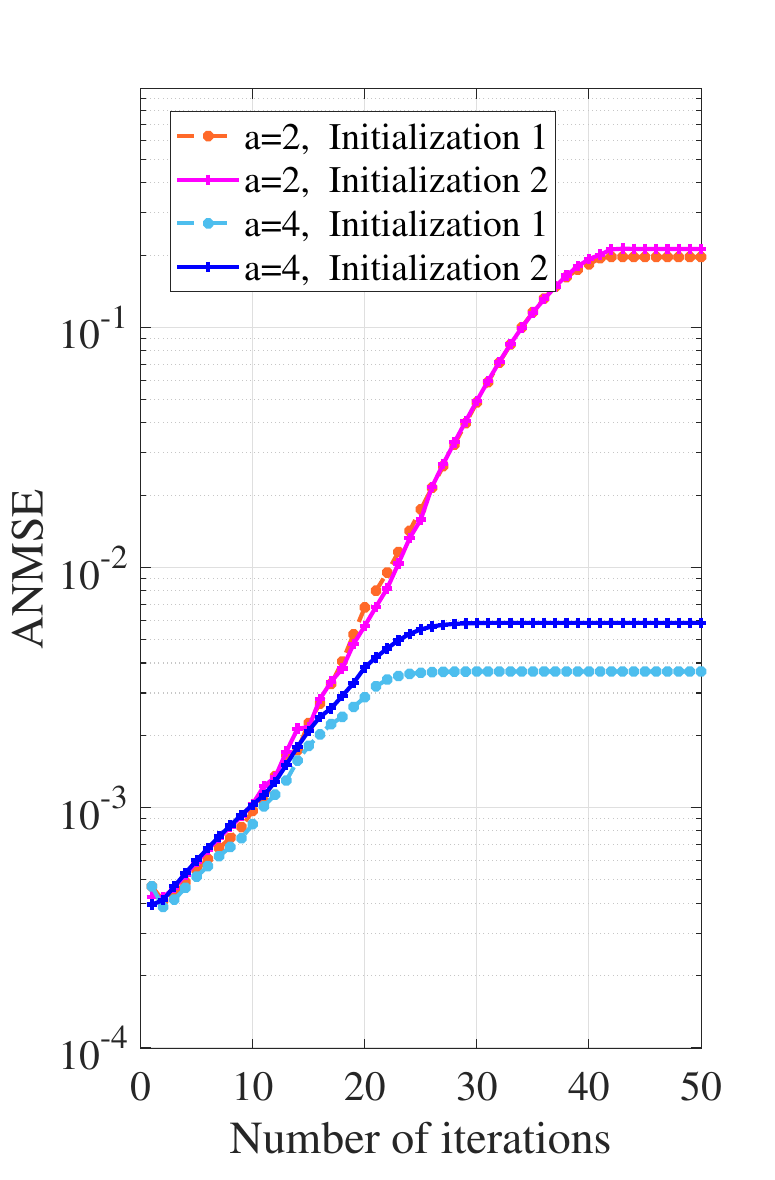}}
        \caption{Convergence behaviors of the PDD algorithm under different initialization schemes.}
        \vspace{-12pt}
        \label{PDD_converve}
    \end{figure}
    
    \begin{figure}[t]
        \centering  
        \includegraphics[width=0.45\textwidth]{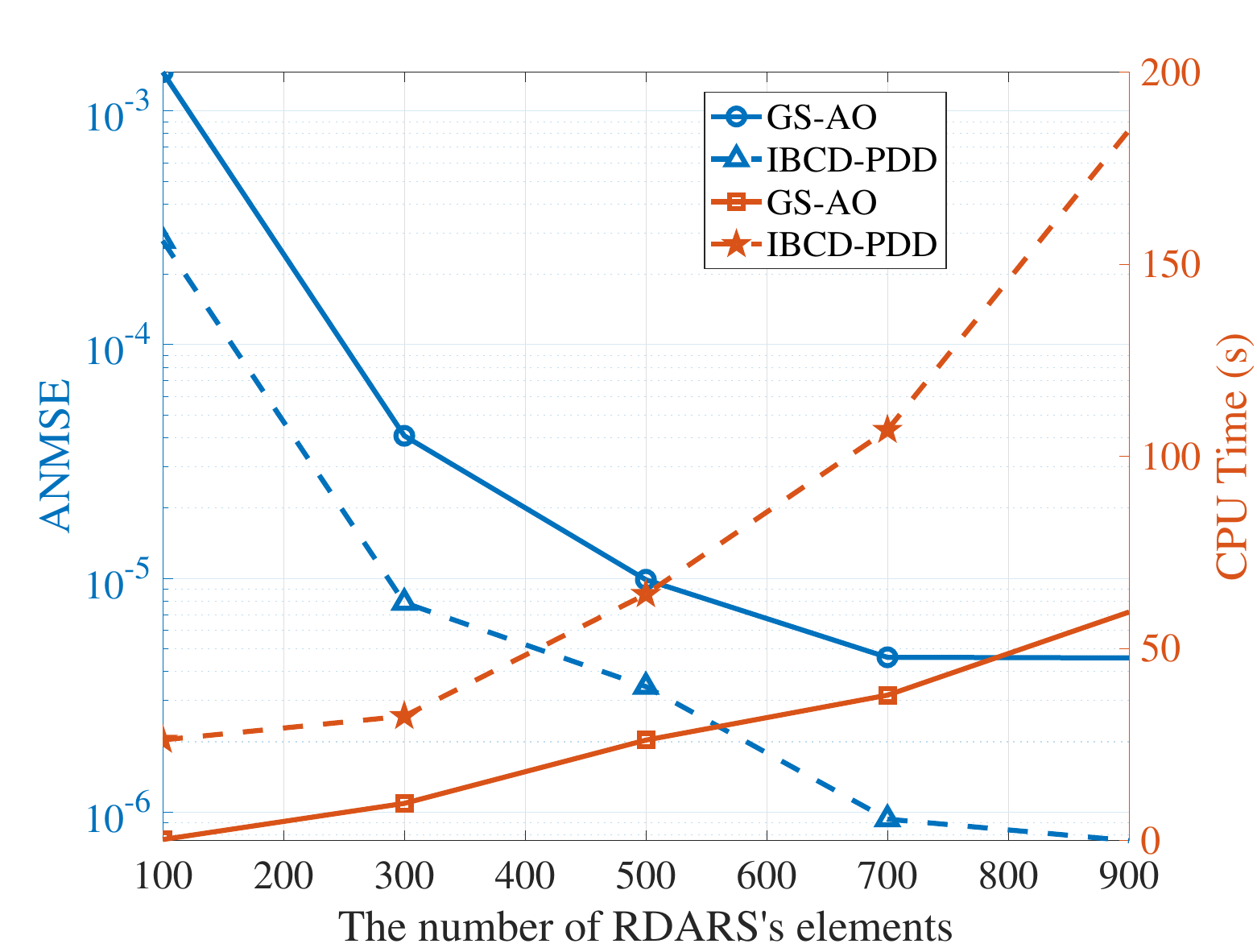}
        \caption{Performance-complexity trade-off between IBCD-PDD and GS-AO methods.}
        \label{tradeoff}
    \end{figure}
    \begin{figure}[t]
        \centering  
        \includegraphics[width=0.45\textwidth]{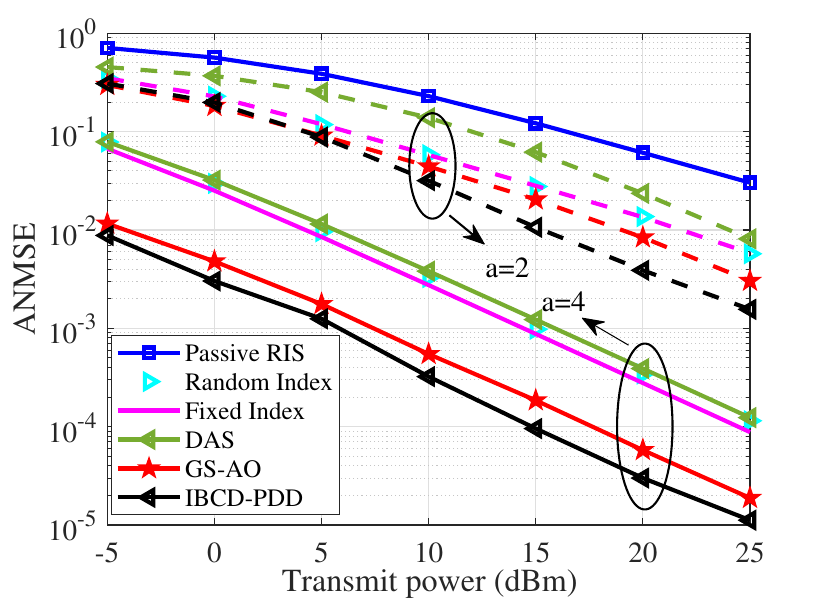}
        \caption{ANMSE vs the transmit power for different algorithm comparisons with $a=2$ or $a=4$ when $N=256$.}
        \vspace{-12pt}
        \label{RDARS_VERSUS_P}
    \end{figure}

{\color{black}    
\subsection{Performance-Complexity Trade-off.}
Fig.~\ref{tradeoff} illustrates the performance-complexity trade-off between the proposed IBCD-PDD and GS-AO methods.
The IBCD-PDD method demonstrates superior ANMSE performance over the GS-AO method as the scale of RDARS increases. 
However, the CPU time required by the IBCD-PDD method increases sharply as the number of RDARS elements $N$ increases, which is always higher than that of the GS-AO method.
This pronounced difference is because the computational complexity of the IBCD-PDD method scales with $N$ to the 3.5th power, in contrast to the quadratic scaling of the GS-AO method. 
Therefore, the superior performance of the IBCD-PDD comes at the cost of increased computational complexity when compared to the GS-AO method. This relationship underscores the necessity for a judicious selection of the method based on the specific requirements and constraints of the application at hand, such as the scale of the RDARS system and the available computational resources.
}

\vspace{-8pt}   
\subsection{ANMSE vs $P$.}
Fig.~\ref{RDARS_VERSUS_P} shows the ANMSE performance of all considered schemes as a function of the transmit power. The RDARS-aided MIMO system is equipped with a total of 256 elements on RDARS, and either 2 or 4 elements operating in the connection mode.
Firstly, it is evident that both the ``Random Index" and ``Fixed Index" schemes outperform the DAS or the passive RIS-aided systems. This observation underscores the inherent advantages of the RDARS's structure, even without the mode selection optimization. Interestingly, the ``Random Index" and ``Fixed Index" schemes exhibit similar performance, which aligns with previous findings in RIS-aided systems employing random and identity phase coefficients \cite{ZhaoxinMSE2022}. Although individual channel realizations may yield varying performance for the random index scheme as compared to the fixed index scheme, such performance fluctuations tend to cancel out in Monte Carlo simulations, resulting in similar overall performance. 
These performance improvements can be attributed to the additional but non-optimized connected links on RDARS.
Furthermore, the proposed ``IBCD-PDD" algorithm is observed to have the best ANMSE performance, demonstrating its effectiveness. Notably, the performance gain of the ``IBCD-PDD" scheme exhibits substantial growth when the number of connected elements increases from 2 to 4. Additionally, the ``GS-AO" algorithm shows promise in approaching a high-quality sub-optimal solution with low complexity.
The performance gap between the proposed ``IBCD-PDD" algorithm and the ``Random Index/Fixed Index" scheme highlights the selection gain originating from the mode selection optimization. This gap becomes more significant as the number of connected elements increases, emphasizing the importance of leveraging the additional DoF offered by the channel-aware placement of elements in the connection mode.
    \begin{figure}[t]
        \centering  
        \includegraphics[width=0.45\textwidth]{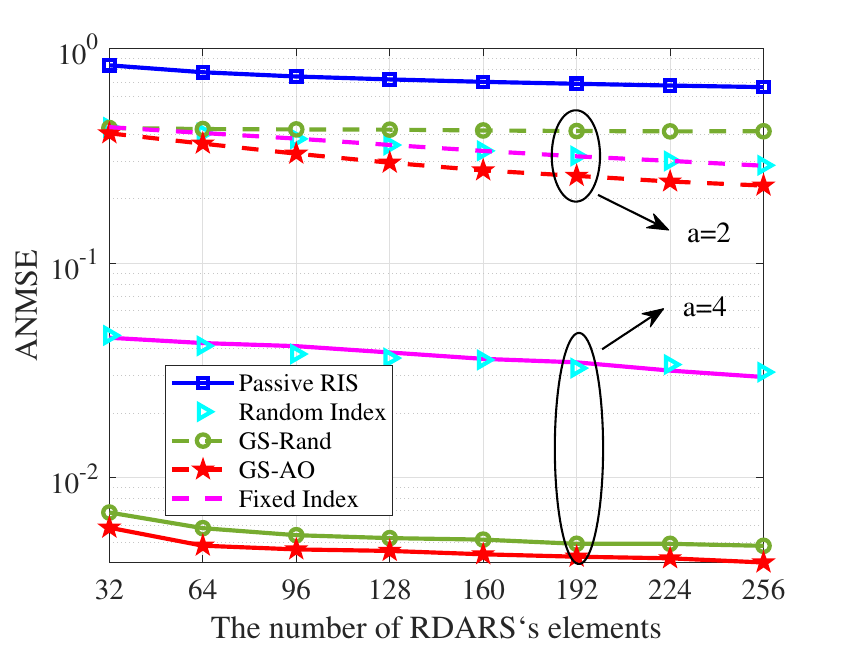}    
        \caption{ANMSE vs the number of RDARS's elements for different algorithm comparisons when $a=2$ and $a=4$.}
        \label{RDARS_VERSUS_N}
    \end{figure}

\vspace{-8pt}
\subsection{ANMSE vs $N$.}
The relationship between the ANMSE performance and the number of RDARS elements for $a=2/4$ under different algorithm setups is illustrated in Fig.~\ref{RDARS_VERSUS_N}. The ``GS-Rand" scheme represents a greedy search-based mode selection scheme without phase optimization.
Firstly, it is worth noting that all RDARS schemes outperform the conventional RIS-aided systems, even with a small number of elements in the connection mode. This superiority is further enhanced as the value of $a$ increases from 2 to 4.
Furthermore, the performance gap between the ``GS-AO" and the ``GS-Rand" schemes represents the reflection gain achieved through phase optimization of the passive elements on RDARS. This gap becomes more prominent as the number of RDARS elements increases when $a=2$. However, when $a=4$, this gap remains relatively stable due to the gradual dominance of the distribution gain over the reflection gain.
In addition, both the ``Random Index" and the ``Fixed Index" schemes outperform the ``GS-Rand" scheme when $a=2$, since the reflection gain takes precedence over the distribution gain and selection gain. However, as the value of $a$ increases to 4, the distribution gain and selection gain start to outweigh the reflection gain.
Finally, the proposed ``GS-AO" algorithm achieves the best performance by effectively integrating the distribution gain, reflection gain and selection gain in both cases.
In summary, these findings underscore the advantages of RDARS systems and emphasize the significance of striking a balance between the distribution gain and reflection gain for optimal performance.

    \begin{figure}[t]
        \centering  
        \includegraphics[width=0.45\textwidth]{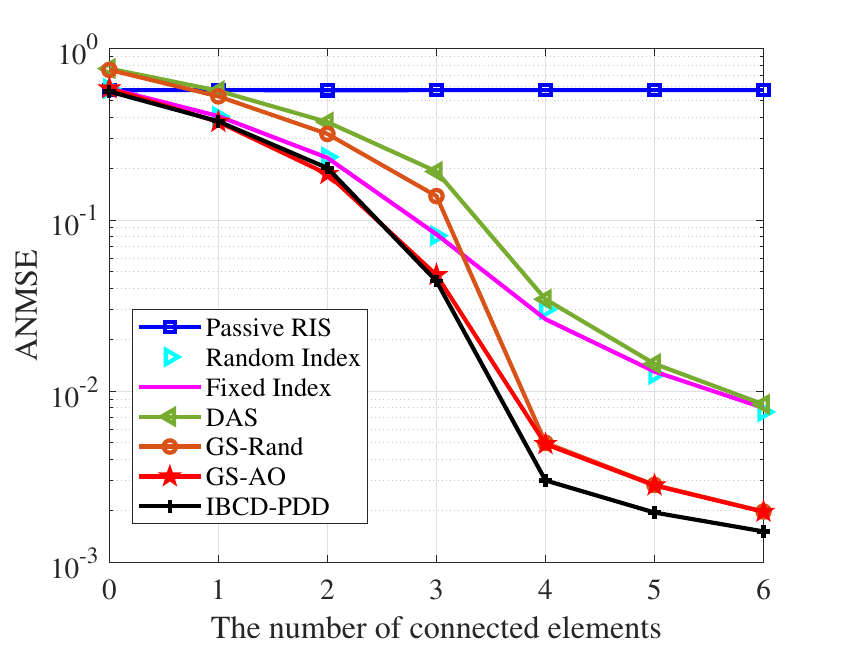}
        \caption{ANMSE vs the number of connected elements for different algorithm comparisons when $N=256$.}
        \vspace{-12pt}
        \label{RDARS_VERSUS_A}
    \end{figure}

        \begin{figure}[t]
        \vspace{-6pt}
        \centering  
        \subfigure[Topology 1]{
        \includegraphics[width=0.24\textwidth]{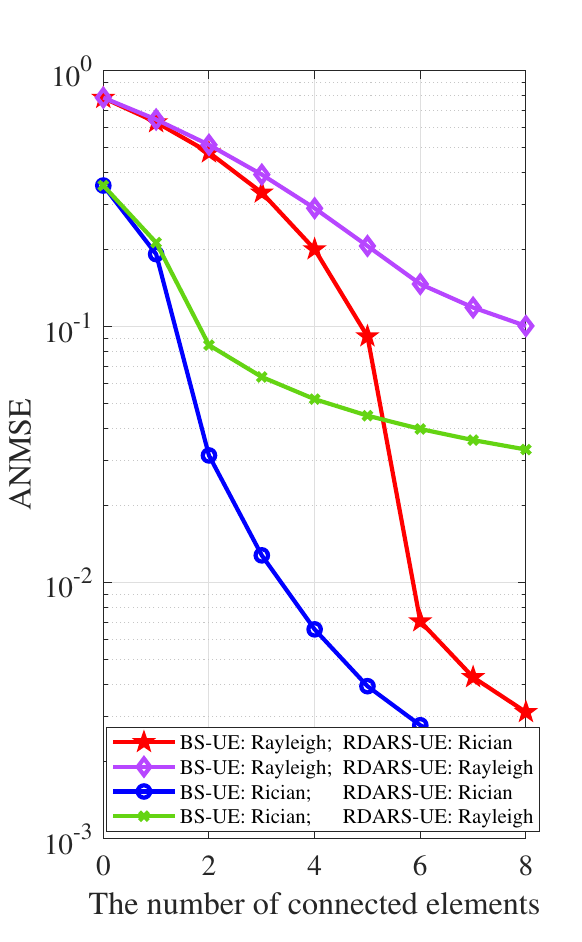}}
        \hspace{-6.5mm}
        \subfigure[Topology 2]{
        \includegraphics[width=0.24\textwidth]{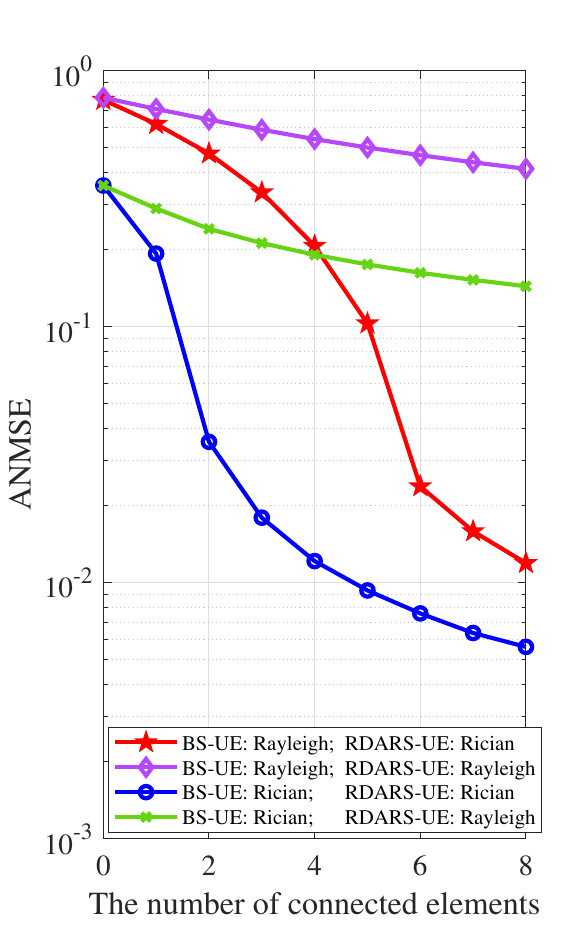}}
        \caption{ANMSE vs the number of connected elements under various channel conditions when $N_r=4,M=6$ in (a) Topology 1 (b) Topology 2.}
        \vspace{-12pt}
        \label{RDARS_TOPO}
    \end{figure}

\vspace{-8pt}
\subsection{ANMSE vs $a$.}
To further harness the distribution gain in the RDARS, we investigate the ANMSE performance versus the number of elements in the connection mode for different algorithms with $N=256$, as depicted in Fig.~\ref{RDARS_VERSUS_A}. 
When $a$ is set as 0, the RDARS-aided systems reduce to the conventional RIS-aided systems, with ``DAS" being equivalent to the traditional MIMO system. In this scenario, all schemes, including ``Passive RIS/GS-AO/Fixed Index/Random Index/IBCD-PDD", outperform the ``GS-Rand" scheme due to the optimization of phase coefficients. ``GS-Rand" performs better than ``DAS" because of the reflection elements.
As $a$ increases, the  ``Passive RIS" scheme remains unchanged, while the other schemes exhibit improving performance due to the distribution gain. Furthermore, as $a$ increases, the mode selection becomes more significant, resulting in a larger performance gap between the ``IBCD-PDD" and ``Fixed Index/Random Index" schemes. Additionally, the performance improvement for the ``IBCD-PDD" algorithm becomes less significant when $a$ exceeds a certain threshold. While it still contributes to the enhancement of ANMSE performance, the gains gradually saturate. 

To determine the suitable value for the number of elements in the connection mode in terms of the hardware complexity and performance gain, we analyze the ANMSE performance versus the number of elements in the connection mode under various system typologies and channel conditions, as shown in Fig.~\ref{RDARS_TOPO}. We consider four channel conditions based on the presence or absence of LoS components in the BS-UE and RDARS-UE channels. {\color{black} In topology 1, the RDARS is located in the middle of the BS and users with $d_R=100$, while in topology 2, it is closer to the BS with $d_R=200$.}
We clearly find that the performance gain of the proposed scheme is not significant as $a$ increases in both topologies when the RDARS-UE channel lacks LoS components. However, when the LoS paths exist between the RDARS and users, the RDARS architecture exhibits substantial performance improvement as $a$ increases.
Furthermore, the slope of MSE performance initially increases and then decreases as the number of elements in the connection mode increases when the deterministic LoS paths exist in the RDARS-UE channel. This suggests an optimal number of connected elements for maximizing the slope of performance improvement exists.
Specifically, when the BS-UE channel follows Rayleigh fading, the slope becomes larger when $a\leq6$ and then decreases when $a>6$, where the transformation threshold is equal to the number of users. However, under the condition of the BS-UE channel subject to Rician fading, this slope changes at $a=2$, which equals the difference between the number of receive antennas at the BS and the number of users. 
In conclusion, the optimal number of elements in the connection mode on RDARS depends on the channel conditions, the number of users, and the number of receive antennas at the BS. These findings provide crucial insights when selecting the appropriate value of connected elements in the practical implementation of RDARS-aided systems.

    \begin{figure}[t]
        \centering  
        \includegraphics[width=0.45\textwidth]{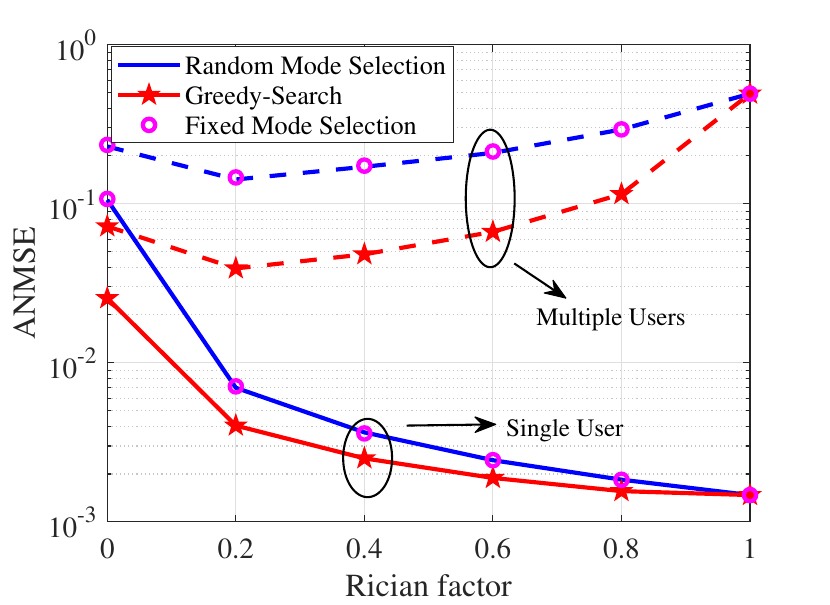}
        \caption{ANMSE vs the  Rician factor under the case of a single user and multiple users for different index selection schemes.}
        \vspace{-12pt}
        \label{RDARS_Ricianfactor}    
    \end{figure}
    
\subsection{ANMSE vs $\kappa_t$.}
In Fig.~\ref{RDARS_Ricianfactor}, we explored the ANMSE performance versus the Rician factor $\kappa_t$ under different mode selection schemes with the same phase coefficients, considering both single-user and multiple-user scenarios. 
In the single-user case, as the Rician factor $\kappa_t$ increases from 0 to 1, the channel transitions from Rayleigh fading to a LoS channel. Consequently, the ANMSE performance is gradually improving. Additionally, as the Rician factor grows, the performance gap between the optimized mode selection scheme and the random/fixed mode selection schemes diminishes. This implies that the optimization of the mode selection becomes less critical as the strength of the LoS component increases.
On the other hand, the ANMSE performance in the case of multiple users does not exhibit a monotonic relationship with the Rician factor due to the increased interference among multiple users as the LoS components become more dominant. The presence of stronger LoS components will lead to more significant interference, impacting the overall performance.
Lastly, when considering only LoS components ($\kappa_T=1$), both the optimized mode selection scheme and the random/fixed mode selection scheme achieve the same ANMSE performance. This finding aligns with Proposition~\ref{pro_LoS}, as stated in Section III, suggesting the mode selection matrix ${\bf{A}}$ can be arbitrarily chosen when only LoS channels are present. 
In conclusion, the ANMSE performance in the RDARS is influenced by the Rician factor, with different trends observed in single-user and multiple-user scenarios. The necessity for mode selection optimization diminishes when only LoS components are present. These insights are valuable for guiding the practical implementation of RDARS systems.

{\color{black} 
\subsection{ANMSE vs $\beta_{c}$.}
Considering such spatially correlated channels, we have studied the ANMSE performance as a function of the number of elements $N$ in the RDARS under various levels of spatial correlation, as illustrated in Fig.~\ref{RDARS_VERSUS_Correlation}.
\begin{figure}[t]
        \centering  
        \includegraphics[width=0.45\textwidth]{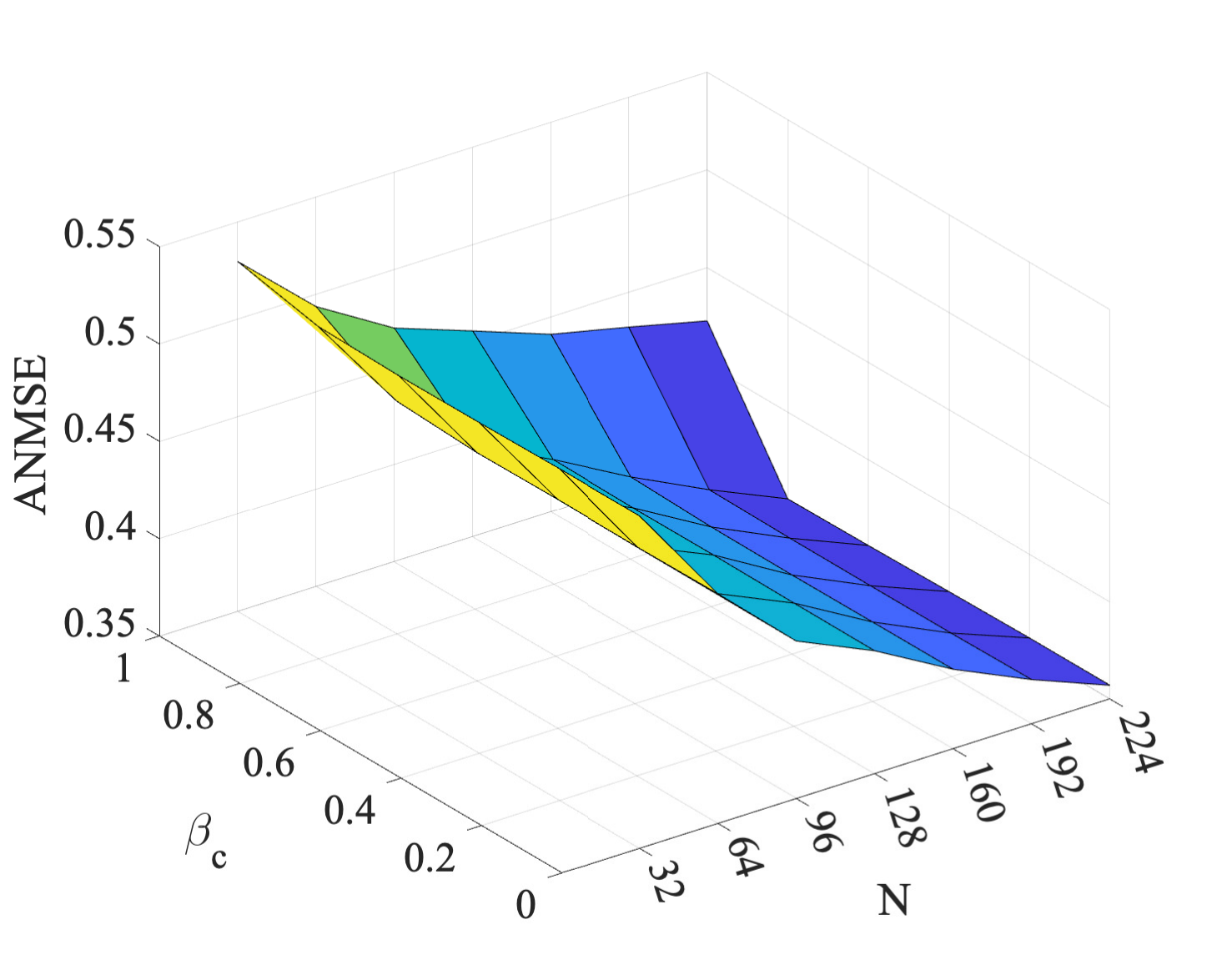}    
        \caption{The ANMSE versus the number of elements $N$ in the RDARS under various levels of spatial correlation.}
        \label{RDARS_VERSUS_Correlation}
        \vspace{-12pt}
    \end{figure}
It can be observed from Figs.~\ref{RDARS_VERSUS_Correlation} that an increase in the number of RDARS elements 
$N$ enhances the MSE performance. However, the spatial correlation at the RDARS negatively affects this performance.
As the correlation coefficient $\beta_{c}$ varies from $0$ to $1$, the channels associated with the RDARS evolve from uncorrelated Rician channels to fully correlated ones.
With an increasing level of spatial correlation, the rank of the channels related to the RDARS decreases. This reduction consequently lowers the spatial multiplexing gain and intensifies multiuser interference, leading to a degradation in performance.}

{\color{black}
\subsection{ANSME vs the RDARS's Position.}
Fig.~\ref{_Distance} illustrates the ANMSE performance as a function of the horizontal distance of the RDARS, i.e., $d_{R}$, under both single-user and multi-user scenarios. 
The conventional passive RIS and DAS are compared as benchmarks against the same three-dimensional setup as depicted in Fig.~\ref{system topologogies}. 
In the single-user system shown in Fig.~\ref{_Distance} (a), the passive RIS achieves optimal performance when located near the user or the BS, while the performance of the DAS deteriorates as it approaches the BS.
The RDARS demonstrates superior performance when in proximity to the user.
However, at a distance of 200 meters, the MSE is lower than at 160 meters when there is one element in connection mode, i.e., $a=1$, due to the dominance of reflection gain over distribution and selection gains.  Conversely, when $a=4$, the distribution and selection gains become dominant, reversing this trend.
In Fig.~\ref{_Distance} (b), the scenario includes four users in the RDARS-aided uplink communication system.  The RDARS shows a consistent performance trend with $a=1$ as the horizontal distance varies, driven by the dominant reflection gain. As $a$ increases to 4, the RDARS performance is poorest near the BS.  Based on these observations, we conclude that better MSE performance is achieved when the RDARS is closer to the users, offering valuable guidelines for the practical deployment of RDARS-aided systems.  

    \begin{figure}[t]
        \centering  
        \subfigure[Single user]{
        \includegraphics[width=0.45\textwidth]{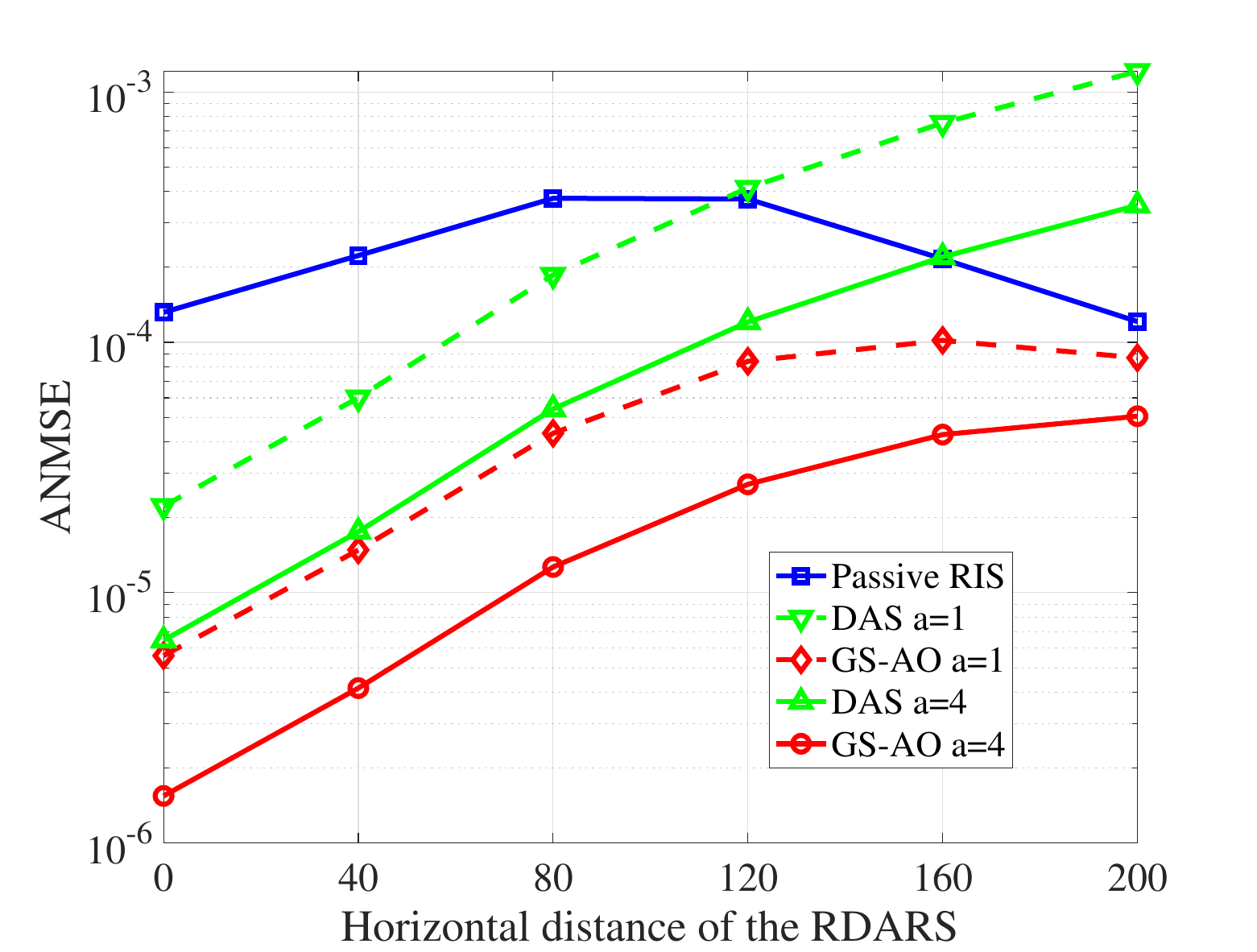}}
        \hspace{-6.5mm}
        \subfigure[Multiple user with $M=4$]{
        \includegraphics[width=0.45\textwidth]{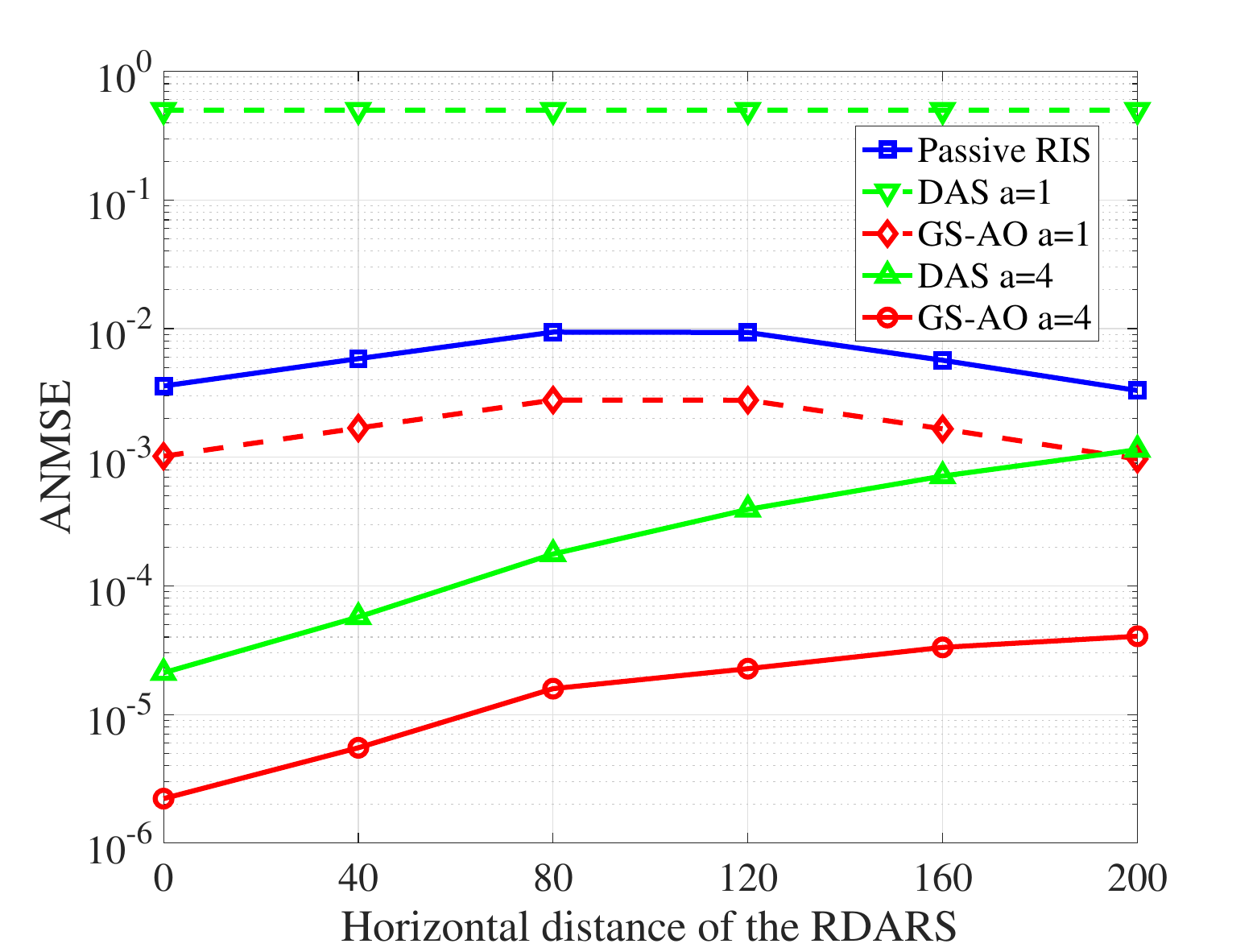}}
        \caption{ANMSE vs the horizontal distance of the RDARS in (a) Single user (b) Multiple users.}
        \label{_Distance}
    \end{figure}

}

\section{Conclusion}
In this paper, we investigated a novel RDARS-aided uplink MIMO communication system and explored the performance limits of this innovative structure.   
To improve the system transmission reliability, the sum MSE minimization problem of all data streams was formulated. 
The flexible and programmable configuration
of RDARS’s elements, where each element can be dynamically switched between the connection mode and the reflection mode, was also exploited to enhance the system performance. 
To address the intractable binary, unit-modulus, and cardinality constraints in the optimization problem, we proposed an inexact BCD-based PDD algorithm. Additionally, we presented a low-complexity greedy-search-based AO algorithm, which provides a near-optimal solution with low computational complexity.
Numerical results demonstrated the proposed architecture showed great superiority in comparison with the conventional RIS and DAS systems. 
We also concluded that the optimal number of elements in the connection mode on RDARS is closely related to the channel conditions, the number of users, and the number of BS antennas. Moreover, some insightful results about the practical implementation of RDARS were investigated. 

\appendix
\subsection{Proof of Proposition 2}
\begin{proof}
 In the single-user SISO case, the MSE minimization problem is equivalently transformed to 
\begin{align}
    \max_{{\bm{\theta}}} ||{{h}}_{d} \!+\! {\bf{g}}^H ({\bf{I}} \!-\! {\bf{A}}) {\rm diag}({\bm{\theta}}) {\bf{h}}_{r} ||^2  +{\bf{h}}_{r}^H {\bf{A}} {\bf{h}}_{r} 
\end{align}
The optimal ${\bm{\theta}}^{\star}$ in the SISO case is derived as:
\begin{align}
     {\bm{\theta}}_{n}^{\star} = \left\{
      \begin{aligned}
      & e^{j (\angle{ {{h}}_{d} }- \angle{{{h}}_{r,n}} + \angle{{{g}}_{n}}) },   \quad n \in \mathcal{\tilde N} ,  \\
     & {\rm any}~{{\theta}}_n~{\rm with}~|{{\theta}}_n|=1,   \quad {\rm otherwise},
      \end{aligned} 
      \right.  
\end{align}
where $\mathcal{\tilde N} \triangleq \{n| {\bf{A}}_{n,n}=0 \}$. For simplicity, the optimal ${\bm{\theta}}^{\star}$ can be selected as ${\bm{\theta}}^{\star}=e^{j (\angle{ {{h}}_{d} }- \angle{{\bf{h}}_{r}} + \angle{{\bf{g}}}) }$.
With the optimal ${\bm{\theta}}^{\star}$, the mode selection optimization turns into 
\vspace{-6pt}
\begin{align}
    \max_{{{x}_n},\forall n} \left(|{{h}}_{d}| \!+\! \sum\limits_{n}^{N}  |{{g}}_{n}| |{{h}}_{r,n}| \!-\! \sum\limits_{n}^{N} x_n |{{g}}_{n}| |{{h}}_{r,n}|  \right)^2 \!+\! \sum\limits_{n}^{N} x_n |{{h}}_{r,n}|^2
\end{align}
Since $|g_n|$ is much smaller than 1, ${x_n}|g_n||h_{r,n}|$ can be ignored.  The maximization problem turns into selecting the $a$ largest $|h_{r,n}|$.
Thus, we complete the proof.
\end{proof}

\vspace{-12pt}
\subsection{Proof of Proposition 3}
\begin{proof}    
Assume the LoS channels ${\bf{G}}=\kappa_g {\bf{a}}_{gr}  {\bf{a}}_{gt}^H$, then the problem turns into
\begin{align}\label{LoS}
    \max_{{\bm{\theta}},{\bf{A}}} \quad ||{\bf{h}}_{d} \!+\! \kappa_g {\bf{a}}_{gt}  {\bf{a}}_{gr}^H ({\bf{I}} \!-\! {\bf{A}}) {\rm diag}({\bm{\theta}}) {\bf{h}}_{r} ||^2  + {\bf{h}}_{r}^H {\bf{A}} {\bf{h}}_{r} 
\end{align}
Denoting ${\bf{\tilde a}}_r=\kappa_g{\bf{a}}_{gr}^H ({\bf{I}} \!-\! {\bf{A}}) {\rm diag}({\bf{h}}_{r})$, the problem can be equivalently turned into
\begin{align}
    \max_{{\bm{\theta}}} \quad  ||{\bf{h}}_{d} \!+\!  {\bf{a}}_{gt}  {\bf{\tilde a}}_r^H   {\bm{\theta}} ||^2  
\end{align}
With some mathematical transformation, it can be rewritten as
\begin{align}
    \max_{{\bm{\theta}}} \quad  {\bm{\theta}}^H {\bf{\tilde a}}_r {\bf{\tilde a}}_r^H   {\bm{\theta}} {\bf{a}}_{gt}^H {\bf{a}}_{gt}   + 2\Re \left\{ {\bf{h}}_{d}^H  {\bf{\tilde a}}_t  {\bf{\tilde a}}_r^H   {\bm{\theta}}  \right\}
\end{align}
The optimal ${\bm{\theta}}^{\star}$ should satisfy 
${\bm{\theta}}^{\star} = e^{j (\angle{{\bf{\tilde a}}_r} + \alpha )  }$ to maximize ${\bm{\theta}}^H {\bf{\tilde a}}_r {\bf{\tilde a}}_r^H   {\bm{\theta}}$ with unit-modulus constraints, where $\alpha$ can be arbitrarily chosen.
On the other hand, to maximize the second term $\Re \left\{ {\bf{h}}_{d}^H  {\bf{ a}}_t  {\bf{\tilde a}}_r^H   {\bm{\theta}}  \right\}$, it can be transformed to
\vspace{-6pt}
\begin{align}
    \max_{{\bm{\theta}}}  \Re \left\{  |{\bf{h}}_{d}^H  {\bf{ a}}_t| \sum\limits_{n}^{N}  |{{\tilde a}}_{r,n}| e^{j (\angle{{{\theta}}_{n}}-\angle{{{\tilde a}}_{r,n}} +\angle{{\bf{h}}_{d}^H  {\bf{ a}}_t} ) }  \right\}
\end{align}
the optimal ${\bm{\theta}}^{\star}$ can be derived as $ {\bm{\theta}}^{\star}=e^{j(\angle{\bf{{\tilde a}}_{r}} - \angle{{\bf{h}}_{d}^H  {\bf{ a}}_t})}$. Combining these two optimal conditions, the optimal ${\bm{\theta}}^{\star}$ for problem \eqref{LoS} can be determined as
$
    {\bm{\theta}}^{\star}=e^{j(\angle{\bf{{\tilde a}}_{r}} - \angle{{\bf{h}}_{d}^H  {\bf{ a}}_t})}.
$.
Based on the derived optimal ${\bm{\theta}}^{\star}$, the mode selection optimization problem under the LoS case is expressed as
\vspace{-6pt}
\begin{align}
    \max_{{{x}_n},\forall n} \quad &{\bf{ a}}_t^H {\bf{ a}}_t |{\bf{{\tilde a}}_{r}}|^2 \!+\! 2|{\bf{h}}_{d}^H  {\bf{ a}}_t|\sum\limits_{n}^{N}  |{{\tilde a}}_{r,n}| + \sum\limits_{n}^{N} x_n |{{h}}_{r,n}|^2 
\end{align}
where
\vspace{-12pt}
\begin{align}
        |{\bf{{\tilde a}}_{r}}|^2 &=  \kappa_g^2 \sum\limits_{n}^{N}  |{{ a}}_{r,n}|^2 |{{ h}}_{r,n}|^2 (1\!-\!x_n)^2, \nonumber \\        
         |{{\tilde a}}_{r,n}| &= \kappa_g |{{ a}}_{r,n}| |{{ h}}_{r,n}| (1\!-\!x_n),
\end{align}
Under the LoS channel ${{\bf h}}_{r}$, we have $|{{ h}}_{r,1}|=|{{ h}}_{r,n}|=\cdots=|{{ h}}_{r,N}|$. 
Thus, it can be readily referred that the mode selection matrix ${\bf{A}}$ can be arbitrarily chosen. Then, we complete the proof. 
\end{proof}

\bibliographystyle{IEEEtran}
\bibliography{RDARS_MSE_R1}

\vfill

\end{document}